%% file: arXiv-version.tex
\documentclass[a4paper,USenglish]{article}
\usepackage{subcaption,here}

\title{The Price of Connectivity Augmentation\\ on Planar Graphs\thanks{The work in this paper originated at the 11th Annual Workshop on Geometry and Graphs held at the Bellairs Research Institute, 8--15 March, 2024.}}

\author{Hugo A. Akitaya\thanks{University of Massachusetts Lowell, USA. \texttt{hugo\_akitaya@uml.edu, frederick\_stock@student.uml.edu}. Research supported in part by NSF grant CCF-2348067.}
\and
Justin Dallant\thanks{Aarhus University, Denmark.\texttt{justindallant@cs.au.dk}}
\and
Erik D. Demaine\thanks{CSAIL, Massachusetts Institute of Technology, Cambridge, MA, USA. \texttt{edemaine@mit.edu}}
\and
Michael Kaufmann\thanks{Wilhelm-Schickard Institut für Informatik, University of Tübingen, Germany. \texttt{michael.kaufmann@uni-tuebingen.de}}
\and
Linda Kleist\thanks{Department of Informatics, Universit\"at Hamburg, Germany. \texttt{linda.kleist@uni-hamburg.de}}
\and
Frederick Stock\footnotemark[2]
\and
Csaba D. T\'oth\thanks{Department of Mathematics, California State University Northridge, Los Angeles, CA; and Department of Computer Science, Tufts University, Medford, MA, USA. \texttt{csaba.toth@csun.edu}.{Research supported in part by the NSF award DMS-2154347.}}
\and
Torsten Ueckerdt\thanks{Karlsruhe Institute of Technology, Germany. \texttt{torsten.ueckerdt@kit.edu}. {Research supported by the German Research Foundation DFG-520723789.}}
}

\let\realbfseries=\bfseries
\def\bfseries{\realbfseries\boldmath}

\newcommand{\opt}{\textsf{OPT}}

\newcommand{\lcr}{\mathrm{lcr}}

\newcommand{\planarsat}{planar~3-SAT}

\usepackage{mathtools}
\usepackage{amsfonts}
\usepackage{amsthm, amsmath}
\usepackage{thmtools}
\usepackage{thm-restate}
\usepackage{xcolor}
\usepackage{fullpage}
\usepackage{hyperref}
	\hypersetup{colorlinks=true, citecolor=blue, linkcolor=red, urlcolor=pink, pdfstartview={FitH}}

\usepackage[numbers,sort&compress]{natbib}

\usepackage{cleveref}

\newtheorem{theorem}{Theorem}
 
\newtheorem{lemma}{Lemma}
\newtheorem{observation}{Observation}
\Crefname{observation}{Observation}{Observations}
\newtheorem{corollary}{Corollary}

\newtheorem{question}{Question}
\newtheorem{proposition}{Proposition}
\newtheorem{claim}{Claim}

\DeclarePairedDelimiter\ceil{\lceil}{\rceil}
\DeclarePairedDelimiter\floor{\lfloor}{\rfloor}

\begin{document}

\maketitle

\begin{abstract}
Given two classes of graphs, $\mathcal{G}_1\subseteq \mathcal{G}_2$, and a $c$-connected graph $G\in \mathcal{G}_1$, we wish to augment $G$ with a smallest cardinality set of new edges $F$ to obtain a $k$-connected graph $G'=(V,E\cup F) \in \mathcal{G}_2$. 
In general, this is the $c\to k$ connectivity augmentation problem. Previous research considered variants where $\mathcal{G}_1=\mathcal{G}_2$ is the class of planar graphs, plane graphs, or planar straight-line graphs. In all three settings, we 
prove that the $c\to k$ augmentation problem is NP-complete when $2\leq c<k\leq 5$. 

However, the connectivity of the augmented graph $G'$ is at most $5$ if $\mathcal{G}_2$ is limited to planar graphs. We initiate the study of the $c\to k$ connectivity augmentation problem for arbitrary $k\in \mathbb{N}$, where $\mathcal{G}_1$ is the class of planar graphs, plane graphs, or planar straight-line graphs, and $\mathcal{G}_2$ is a beyond-planar class of graphs: 
$\ell$-planar, $\ell$-plane topological, or $\ell$-plane geometric graphs. We obtain tight bounds on the tradeoffs between the desired connectivity $k$ and the local crossing number $\ell$ of the augmented graph $G'$. We also show that our hardness results apply to this setting.

The connectivity augmentation problem for triangulations is intimately related to edge flips; and the minimum augmentation problem to the flip distance between triangulations. We prove that it is NP-complete to find the minimum flip distance between a given triangulation and a 4-connected triangulation, settling an open problem posed in 2014, and present an EPTAS for this problem. 
\end{abstract}

\input{intro}
\input{warmup}

\input{topological}

\input{trees}

\section{Hardness Results}
\label{sec:hardness}

In this section we prove our main hardness results (\cref{thm:hardness} and \cref{cor:flip-hardness}). 
We reduce from \planarsat. 
A 3-SAT instance is given by a boolean 3-CNF formula $\varphi$, with $n$ variables $x_1, x_2, \dots, x_n$ and $m$ clauses $C_1, C_2, \dots, C_m$.
In \planarsat, $\varphi$ can be represented as a planar bipartite graph with vertices for every variable and clause, and an edge $(x_i,C_j)$ if $x_i$ or $\neg x_i$ is a literal in clause $C_j$. 
This is a well studied SAT variant proven NP-complete by Lichtenstein~\cite{Lichtenstein1982PlanarFA}.
We have three reductions that respectively cover the $2\rightarrow 3$, $3\rightarrow 4$ and $4\rightarrow 5$ augmentations.
Note that if the $c\rightarrow k$ augmentation problem is NP-complete, so is the $c'\rightarrow k$ problem for $c'<c$ because every $c$-connected graph is also $c'$-connected. 
All three reductions are similar: constructed from three gadgets (variable, clause, and literal) that are put together in the standard way.
\subsection{\texorpdfstring{$2\to 3$}{2 to 3} Connectivity Augmentation is NP-Complete}

We now show that minimally augmenting a planar graph from 2 to 3 connectivity is NP-hard. We first present our variable gadget in \Cref{fig:2-3Variable}. 
Our reduction works for all settings (abstract, topological and geometric) since the only choice we have in embedding the resulting graph occurs locally in each gadget, and changing the embedding (compared to the ones shown in the figures) would result in an equal or worse solution.
This gadget is built by a connected sequence of ``diamonds'' (shown in \Cref{fig:2-3BasicM}) which are two triangles with a shared edge. 
These diamonds have four vertices two of which are degree three and two are degree two. 
We form a variable gadget by connecting diamonds by their degree two vertices (\Cref{fig:2-3Variable}). 
We call these the \emph{boundary vertices} of each diamond.
Without any additional edges each pair of boundary vertices form a $2$-cut where they would disconnect the other two vertices of the diamond. 
We call these two vertices the \emph{central vertices} of each diamond.

\begin{restatable}{lemma}{23Variable} \label{lemma:2-3Variable}
    The variable gadget in \Cref{fig:2-3Variable} composed of $2v$ diamonds can be minimally augmented from 2 to 3 connectivity with $v$ edges.
\end{restatable}

\begin{proof}
    Each pair of central vertices can be disconnected by removing their boundary vertices, then one vertex in every pair of central vertices must be connected to some vertex other than their boundary vertices. 
    There are $2v$ pairs of central vertices each of which needs one additional edge.
    Therefore, a perfect matching between pairs of central vertices would produce an augmenting set with only $v$ edges.
    By construction there are two such sets, shown in red and blue in \Cref{fig:2-3Variable}. 
    We call one matching the positive matching (the blue edges) and the other matching negative (red edges).
    Note $v$ edges are necessary to augment a variable gadget.
    Each augmenting edge can only augment at most two pairs of central vertices. So any augmenting set of less than $v$ edges will leave at least one set of central vertices unmodified and hence its boundary vertices will still form a $2$-cut.
    Therefore $v$ edges are necessary and sufficient to augment a variable gadget with $2v$ diamonds.
\end{proof}

\begin{figure}[ht]
    \centering

    \begin{subfigure}{0.2\textwidth}
        \centering
        \includegraphics[scale=3]{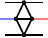}
         \subcaption{}
         \label{fig:2-3BasicM}
    \end{subfigure}
    \hfil
    \begin{subfigure}{0.6\textwidth}
        \centering
        \includegraphics[scale=3]{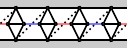}
        \subcaption{}
        \label{fig:2-3Variable}
    \end{subfigure}
    \hfill
    \caption{(a) A diamond and (b) a variable gadget.}
    \label{fig:2-3variableFig}
\end{figure}

We now discuss the literal gadget
which serves as a connector between clause and variable gadgets. 
Intuitively, it can be thought of as a gadget we ``insert'' into a variable gadget to expose an edge that can be made adjacent to a clause gadget.
The literal gadget is shown in \Cref{fig:2-3LiteralM}. 
It is essentially just an extension of the basic diamond from \Cref{fig:2-3BasicM}.
It is a diamond with four additional vertices, such that each edge between a central vertex and a boundary vertex is also the hypotenuse of a triangle formed with an added vertex.
We call these four vertices \emph{auxiliary vertices}.
In our final construction this gadget will be inserted with a pair of diamonds on either side. 
One pair will consist of two diamonds that are a part of a variable gadget (we call these \emph{variable diamonds}) and the pair on the opposing side will consist of two diamonds that are a part of a clause gadget (we call these \emph{clause diamonds}).

\begin{figure}[ht]
    \centering

    \begin{subfigure}{0.3\textwidth}
        \centering
        \includegraphics[width=\textwidth]{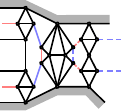}
         \subcaption{}
         \label{fig:2-3LiteralM}
    \end{subfigure}
    \hfill
    \begin{subfigure}{0.6\textwidth}
        \centering
        \includegraphics[width=\textwidth]{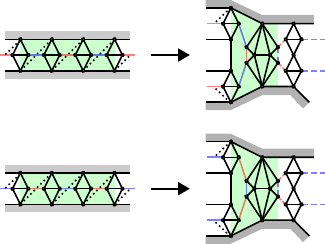}
        \subcaption{}
        \label{fig:2-3LiteralFull}
    \end{subfigure}
    
    \caption{(a) A literal gadget for $2\to 3$ connectivity augmentation. Diamonds on the left are variable diamonds and those on the right are clause diamonds. Solid lines represent potential edges used to augment a variable gadget, dashed lines are possible augmenting edges from a clause gadget. Blue lines represent a true assignment and red represent a false assignment. (b) Literal gadget are inserted into a variable gadget. A positive literal is shown on top, a negated literal is on the bottom.}
    \label{fig:2-3LiteralGadgets}
\end{figure}

\begin{restatable}{lemma}{23Literal}
    The literal gadget can be minimally augmented from 2 to 3 connectivity with exactly 3 edges.
\end{restatable}
\begin{proof}
    
    This gadget has two types of $2$-cuts. One is the same as the diamond, the two boundary vertices form a cut set. 
    As each of the auxiliary vertices are adjacent to the central vertices, as long as at least one is adjacent to a vertex outside of the gadget, the first $2$-cut will be resolved.
    The other type of $2$-cut is formed by each pair of adjacent boundary and central vertices.
    This cut will disconnect the added vertex adjacent to both. 
    Therefore, the four auxiliary vertices must have at least one edge connecting them to some other vertex.

    We now observe that two edges is not sufficient to augment this gadget. As each auxiliary vertex requires an additional edge, the only option would be to connect pairs of auxiliary vertices with these two edges. However these edges only connect vertices of the literal gadget to each other so the boundary vertices of the literal gadget still form a $2$-cut. 
    Hence, at least one augmenting edge cannot have auxiliary nodes for both its end points and so we need at least three augmenting edges. 
    By that argument then three edges is sufficient to augment this gadget. One could just match auxiliary vertices with two edges and then add one more edge connecting an auxiliary vertex to a vertex of either a variable or clause diamond. 
\end{proof}

We will show later that due to features of the larger construction. There will be only two possible minimal augmentations of a literal gadget. 
These two augmentations are shown in \Cref{fig:2-3LiteralM} by blue and red edges. If the gadget is augmented by the blue set of edges, we say that the literal is \emph{true}, if augmented by the red edges we say it is \emph{false}.

Our final gadget is a clause gadget (\Cref{fig:2-3ClauseM}). This gadget consists of three pairs of diamonds (clause diamonds) with their boundary vertices connected such that they border an enclosed region. Inside of which there are two triangles each with a vertex on the boundary of this region, and who have a shared vertex. We call these the \emph{value triangles} of the clause. 
There is a literal gadget adjacent to each pair of clause diamonds. 
We may call these adjacent literals.
With these gadgets we now prove our main theorem.

\begin{restatable}{lemma}{23Clause}\label{lemma:2-3Clause}
    The clause gadget in \Cref{fig:2-3ClauseM} can be minimally augmented with 7 edges if and only if at least one adjacent literal is true. Otherwise 8 are necessary.
\end{restatable}
\begin{proof}
    Each diamond can either be connected to an auxiliary vertex of a literal gadget or a vertex of a value triangle.
    Clearly the best edges to add would be between auxiliary vertices and the clause diamond as this augments both the clause diamonds and the auxiliary vertices of the literal. 
    In fact if a literal is false, these edges must be added or else the auxiliary vertices will be only degree 2.
    However, if a literal is true, these edges are not necessary, this can be observed by the literal at the bottom of \Cref{fig:2-3ClauseM}.
    Further, note, the value vertices of the clause both need to be connected to a diamond of a literal gadget.
    Therefore, if every literal is false, a total of eight edges are required to augment a clause gadget. Two edges for each literal and two edges for the value vertices of the clause gadget.
    However, if at least one literal is true, the two auxiliary vertices can be connected to each other. Therefore at most seven edges are necessary, one for the true literal, two for each false literal and two for the value vertices of the clause gadget. 
    Note that at least seven are always necessary.
    If more than one literal is true, then it is actually more efficient to still connect each auxiliary vertex to a diamond instead of each other. This only uses two edges to augment both diamonds and the auxiliary vertices. Connecting the two auxiliary vertices would result in three edges.
\end{proof}

\begin{figure}[ht]
    \centering

    \begin{subfigure}{0.4\textwidth}
        \centering
        \includegraphics[width=\textwidth]{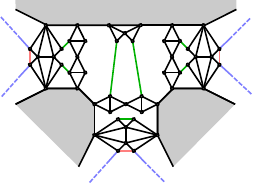}
         \subcaption{}
         \label{fig:2-3ClauseM}
    \end{subfigure}
    \hfill
    \begin{subfigure}{0.5\textwidth}
        \centering
        \includegraphics[width=\textwidth]{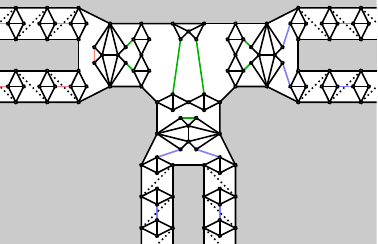}
        \subcaption{}
        \label{fig:2-3ClauseEx}
    \end{subfigure}
    \hfill
    \caption{(a) The clause gadget and (b) the connection between the clause gadget and literal and variable gadgets. Black dotted lines will be explained later.}
    \label{fig:2-3ClauseGadgets}
\end{figure}

\begin{restatable}{lemma}{23NPC} \label{lem:2-3NPC}
    Given a 2-connected planar graph $G=(V,E)$ and an integer $\tau$, it is NP-hard to decide whether there is a set $F$ of edges with $|F|\le \tau$ such that $G'=(V,E\cup F)$ is a 3-connected planar graph.
\end{restatable}
\begin{proof}
    We reduce from \planarsat, in particular given a \planarsat formula $\varphi$ with $n$ variables and $m$ clauses we construct a 2-connected planar graph $G$ and ask if it can be planarly augmented to 3-connectivity with at most $\tau$ edges. 

    We construct $G$ as follows: Take the planar embedding $P$ of $\varphi$. For any variable $x_i$ of $\varphi$, in $P$ $x_i$ and its adjacent edges form a star graph. Take an Eulerian circuit of this star and take its planar embedding (allowing curved arcs). We create a variable gadget for $x_i$ using an even number of diamonds embedded according to the Eulerian Circuit of $x_i$'s star. 
    Each clause of $P$ is replaced by a clause gadget, and where each clause is adjacent to a variable we replace two diamonds of the variable gadget with a literal gadget (see \Cref{fig:2-3LiteralFull}.
    We want to ensure that there is a unique embedding for $G$, except for flipping diamonds. 
    Note, every gadget presented has gray border.
    Following this boundary, we take each pair of sequential vertices and add two extra vertices and create a $k_4$ within this gray region. This joins every two sequential vertices in a 3-block, a minimally 3-connected graph of 4 vertices. 
    Therefore, the clockwise order of these ``boundary vertices'' cannot be modified as the edges of these 3-blocks will cross violating the planarity of $G$.
    We can now compute a value for $\tau$. This will be half the number of diamonds present in every variable gadget (by \Cref{lemma:2-3Variable}) plus 7 for each clause gadget (so $7m$ by \Cref{lemma:2-3Clause}). Note, the number of edges required to augment a literal gadget is equal to the number of edges that would have been required to augment the diamonds it replaced. Hence, literal gadgets are unnecessary for this accounting. Therefore by the validity of \Cref{lemma:2-3Variable,lemma:2-3Clause} there will only be a solution if every clause has at least one adjacent true literal gadget.
    If there is a set of $ \le \tau$ edges that augments $G$, then we can find a satisfying assignment for $\varphi$ by inspecting the assignment of each variable gadget in the augmented graph $G'$. By construction this will mean each clause in $\varphi$ will have at least one true literal, hence $\varphi$ is satisfied.
    
\end{proof}

\begin{restatable}{corollary}{23topo} \label{cor:23-topo}
    \Cref{lem:2-3NPC} still holds if $G$ is a topological (or geometric) plane graph and if $G'$ is required to be a compatible topological (resp., geometric) plane or 1-plane graph.
\end{restatable}

Note, every argument that was made was independent of the magnitude or drawing of any edge (as long as they are planar). Hence our reduction easily extends to the topological and geometric settings, by just embedding each gadget (and therefore $G$) as depicted. The reduction is also easily adapted to the 1-planar setting. In each of our gadgets, we included dashed lines, like the ones between every diamond in the variable gadget. By replacing these with an edge in $G$, augmenting $G$ will necessarily result in a 1-planar graph.

\begin{restatable}{corollary}{23DegCor} \label{cor:2-3-deg}
    Given a min-degree-2 planar graph $G=(V,E)$ and an integer $\tau$, it is NP-hard to decide whether there is a set $F$ of edges with $|F|\le\tau$ such that $G'=(V,E\cup F)$ is a min-degree-3 planar graph.
\end{restatable}

We can use the same reduction with a slight tweak to show that augmenting a planar graph with min-degree-2 to a planar graph with min-degree-3 is NP-complete as well. 
Note the reduction as it stands will not work, as the variable gadgets are already min-degree-3.
However, if we simply contract the edge of every diamond in the reduction we resolve this issue. Variable gadgets will need to augmented just as they were in the connectivity case, one edge for every two contracted diamonds.
Note, we could not use these contracted diamonds for our original reduction, as a 2-cut would occur in a literal gadget assigned false. This can be seen in the top left corner of \Cref{fig:2-3ClauseEx}. 
Therefore we have \Cref{cor:2-3-deg}.

\subsection{\texorpdfstring{$3\to 4$}{3 to 4} Connectivity Augmentation is NP-Complete}

Our reduction works for all settings (abstract, topological and geometric) since 3-connected graphs have a unique combinatorial embedding and the reduction can be embedded geometrically.
We use three gadgets, a variable, a literal, and a clause gadget. 
We ensure that the entire construction is 4-connected except for some select degree-3 vertices.
The variable gadget is shown in~\Cref{fig:3-4Wire-500}.
The ``building block'' of our construction is essentially a wheel graph on four vertices $W_4$.
We modify this $W_4$, removing two edges so that one vertex is degree-3, one is degree-1, and two are degree-2 (see \Cref{fig:mod-k4-500}). 
We call the degree three vertex the \emph{central} vertex and the others the \emph{boundary} vertices.
We construct a variable gadget by ``gluing'' an even number of these modified $W_4$'s together so any two consecutive $W_4$ have two shared boundary vertices.
Hence every boundary vertex is identified with two separate $W_4$'s, while the central vertices are unique to each $W_4$, so every central vertex is precisely degree three and must therefore be augmented.

\begin{figure}[ht]
        \centering

        \begin{subfigure}{0.2\textwidth}
            \centering
            \includegraphics[scale=2.7]{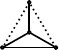}
             \subcaption{}
             \label{fig:mod-k4-500}
        \end{subfigure}
        \hfill
        \begin{subfigure}{0.7\textwidth}
            \centering
            \includegraphics[scale=2.7]{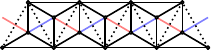}
            \subcaption{}
            \label{fig:3-4Wire-500}
        \end{subfigure}
        \hfill
        \caption{(a) Modified $W_4$, removed edges indicated with dashed line. (b) A variable gadget. Red segments indicate a false assignment, blue indicate a true assignment.}

        \label{fig:3-4variable-500}
    \end{figure}

\begin{restatable}{lemma}{threefourVariable}\label{lemma:3-4Variable}
    The variable gadget (Figure~\ref{fig:3-4Wire-500}) constructed with $2k$ modified $W_4$'s can be minimally augmented with $k$ edges, and there are two such augmentations. 
\end{restatable}
\begin{proof}
    As each central vertex is still unique to its $W_4$, the boundary vertices of each $W_4$ form a 3-cut disconnecting their central vertex.
    Therefore, augmenting this gadget to be 4-connected requires adding an edge from each central vertex to a vertex in a separate $W_4$. This augmentation must be planar and so each central vertex can only be connected to the central vertices of its two adjacent $W_4$'s.
    Note, once a central vertex is connected to another central vertex, each of their adjacent boundary vertices no longer form a $3$-cut. 
    So we only need to connect every other pair of central vertices, and in fact this would be the minimal augmentation for this gadget. 
    Each central vertex must be connected to at least one other vertex, and they can only be connected to other central vertices. 
    Therefore connecting each central vertex to precisely one central vertex must be minimal, and clearly the number of edges required is precisely half the number of central vertices.

    Note there are two unique ways to do this. 
    Arbitrarily identify one $W_4$ as the ``initial'' and orient the variable cycle from this $W_4$.
    The two possible augmentations are therefore either each central vertex is attached to its clockwise neighbor, or each central vertex is adjacent to its counter clockwise neighbor. 
    If clockwise neighbors are adjacent we say the variable is set to ``true'' and we say it is ``false'' if counter clockwise neighbors are connected.
\end{proof}
A literal gadget (\Cref{fig:3-4Literal-500}) is composed of six vertices, and similar to the modified $W_4$'s used in the variable gadget. 
Again, as with the variable gadget, the three boundary vertices form a 3-cut that disconnects the three central vertices. 
As depicted in \Cref{fig:3-4literal-conversion-500} literal gadgets are ``inserted'' into a variable gadget, essentially replacing two of $W_4$'s of a variable.
Consider a literal gadget inserted into a variable $x_i$.
There are only two possible locally minimal edge augmentations of the literal gadget, as depicted in \Cref{fig:3-4literal-conversion-500}.
For a positive literal (\Cref{fig:3-4literal-conversion-500}(top)) if $x_i$ is true (resp., false) we connect the degree-3 vertices with the blue (resp., red) matching. 

\begin{figure}[H]
        \centering

        \begin{subfigure}{0.25\textwidth}
            \centering
            \includegraphics[width=\textwidth]{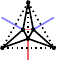}
             \subcaption{}
             \label{fig:3-4Literal-500}
        \end{subfigure}
        \hfill
        \begin{subfigure}{0.65\textwidth}
            \centering
            \includegraphics[width=\textwidth]{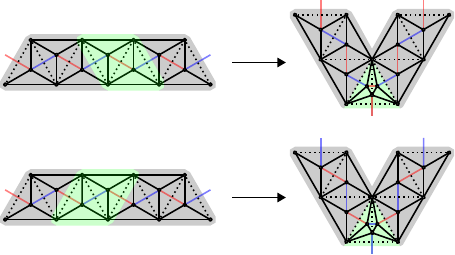}
            \subcaption{}
            \label{fig:3-4literal-conversion-500}
        \end{subfigure}
        \caption{(a) A positive literal gadget (b) How a literal gadget is ``inserted'' into a variable gadget. Two $W_4$'s of the variable are replaced by one literal. A clause gadget is adjacent to the exposed dotted line. On the top is a positive literal, the bottom shows a negated literal.}
        \label{fig:3-4literalGadgets-500}
    \end{figure}



\begin{figure}[H]
        \centering
        \begin{subfigure}{0.3\textwidth}
            \centering
            \includegraphics[width=\textwidth]{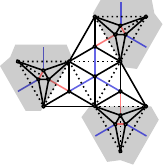}
            \subcaption{}
            \label{fig:3-4Clause-cropped-500}
        \end{subfigure}
        \hfill
        \begin{subfigure}{0.6\textwidth}
            \centering
            \includegraphics[width=.7\textwidth]{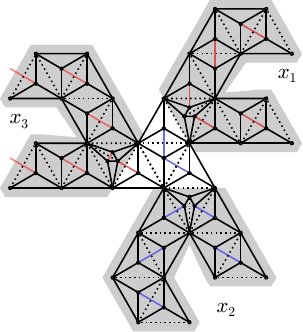}
            \subcaption{}
            \label{fig:3-4Clause-Full-500}
        \end{subfigure}
        \caption{(a) A clause gadget (b) Example where two variables are set to true and one is set to false. Variables are highlighted in grey.}
        \label{fig:3-4clause-500}
\end{figure}

The clause gadget (\Cref{fig:3-4clause-500})
requires 3 new edges, and is achievable iff a literal is true.

\begin{restatable}{lemma}{threefourClause}\label{lemma:3-4Clause}
    %
    The clause gadget shown in \Cref{fig:3-4clause-500} can be augmented with exactly three edges if and only if one adjacent literal gadget is true. Otherwise 4 edges are necessary. 
\end{restatable}
\begin{proof}
    
    Let $c$ be a clause of $\varphi$ on three variables $x_i$, $x_j$, and $x_k$.
    For each literal of $c$ we ``glue'' a term $K_4$ to the corresponding variable gadget.
    We take two adjacent $W_4$'s of the variable gadget and replace them with a literal gadget.
    (\Cref{fig:3-4literal-conversion-500}).
    
    To augment the clause gadget, for each of its $W_4$'s the central vertex must be connected to a vertex of another $W_4$. 
    Note, the central vertex of the value $W_4$ can only connect to one of the central vertices of the terms. While each term $W_4$ can either have its central vertex connected to a vertex in an adjacent literal gadget or to the central vertex of the value $W_4$. Therefore, augmenting this gadget requires at least three edges, one edge for each term.
    We now show that this can be augmented exactly three edges only if at least one adjacent literal gadget is true.
    
    If a term is adjacent to a false literal then the only possible minimal augmentation is to connect the central vertex of the term to a central vertex of the literal gadget. This would augment both the term and the literal by adding only one edge. 
    Any other augmentation would require at least two edges, we would have to connect a central vertex of the literal to a variable $W_4$ and we would have to connect the central vertex of the term to the central vertex of the value of the clause.

    Note that a true literal gadget is already 4-connected. Therefore, if a term $W_4$ is adjacent to a true literal then we can connect its central vertex to the central vertex of the value of the clause, augmenting both to be 4-connected. 

    Therefore we can augment a clause gadget with only 3 edges if and only if it is adjacent to a true literal gadget. Otherwise we require 4 edges. 
    If every term is adjacent to a false literal, then 
    each of the false literal gadgets require one edge and the value of the clause requires an edge.
    As the augmentation must be planar, these edges are disjoint, and so we need to add four edges.
    So, $c$ can be augmented with three edges if and only one adjacent literal gadget is true.
\end{proof}

With these three gadgets we can reduce from \planarsat, giving the following lemma.

\begin{restatable}{lemma}{lem34NPC} \label{lemma:3-4NPC}
    Given a 3-connected planar graph $G=(V,E)$ and an integer $\tau$, it is NP-hard to decide whether there is a set $F$ of edges with $|F|\le \tau$ such that $G'=(V,E\cup F)$ is a 4-connected planar graph.
\end{restatable}

\begin{proof}

    Recall that we consider the decision version of this problem. I.e. given a 3-connected plane graph $G$, can at most $\tau$ edges be added to form a 4-connected plane graph $G'$?

    Given a \planarsat\ formula $\varphi$ with $n$ variables and $m$ clauses. Let $P$ be a planar embedding of $\varphi$.
    We now construct an instance of 3-4 graph augmentation, starting with a 3-connected plane graph.
    
    Take every variable's embedding in $P$ each will form a star graph. For each variable we create a corresponding variable gadget by taking a Eulerian circuit of each of these stars and replacing it with a variable gadget. These gadgets are cycles of an even number of modified $W_4$'s. 
    Next replace every clause node with a clause gadget and for each of the three literals insert corresponding literal gadgets into the relevant variable gadgets, replacing two $W_4$'s. We can then compute what the correct value of $k$ is using \Cref{lemma:3-4Variable,lemma:3-4Clause}. That is, $k$ is half the number of $W_4$'s that occur in all variable gadgets plus $3$ for each clause gadget ($3m$). 
    
    Now given a satisfying assignment $S$ for $\varphi$ we can augment each variable gadget accordingly and then if $S$ is a satisfying assignment each clause gadget should have an adjacent true literal gadget. Therefore by \Cref{lemma:3-4Clause} each clause can be augmented with $3$ edges. 
    So $G$ can be augmented with at most $\tau$ edges.
    
    If there is a set of $\tau$ edges, where $\tau$ is half the number of variable $W_4$'s plus 3 times the number of clauses that planarly augments $G$. Then there is a corresponding satisfying assignment of $\varphi$. 
    By \Cref{lemma:3-4Variable} there are only two ways to minimally augment each variable gadget, corresponding to false or true assignments respectively. Therefore, each variable gadget will uniquely encode a corresponding variable assignment, this gives some possible solution $S'$ for $\varphi$. By \Cref{lemma:3-4Clause} we will only be able to augment each clause gadget with three edges if one of its adjacent literal gadgets is true, otherwise at least one will require 4 edges. So we will only have an augmenting set of $\tau$ edges if every clause of $\varphi$ has at least one true literal. 
    So if there is an augmenting set of $\tau$ edges we can get a satisfying assignment $S'$ for $\varphi$ by checking each variable gadget.

    We conclude that $\varphi$ has a satisfying assignment $S$ if and only if $G$ can be augmented with at most $\tau$ edges.
\end{proof}

\begin{restatable}{corollary}{34Topo}
    
    \label{cor:34-topo}
    \cref{lemma:3-4NPC} still holds if $G$ is a topological (or geometric) plane graph and if $G'$ is required to be a compatible topological (resp., geometric) plane or 1-plane graph.
\end{restatable}

Note, every argument that was made was independent of the magnitude or drawing of any edge (as long as they are planar). Hence our reduction easily extends to the topological and geometric settings, by just embedding each gadget (and therefore $G$) as depicted. The reduction is also easily adapted to the 1-planar setting. In each of our gadgets, we included dashed lines, like the ones between every $W_4$ in the variable gadget. By replacing these with an edge in $G$, augmenting $G$ will necessarily result in a 1-planar graph.

\begin{restatable}{corollary}{34Flip}
    Given a combinatorial triangulation $T$ and an integer $\tau$. Determining if there is a series of $\le \tau$ flips that transforms $T$ into a $4$-connected triangulation $T'$ is NP-complete. 
\end{restatable}

This is a simple extension of our main reduction. Note, every black dotted line corresponds in our gadgets (\Cref{fig:3-4variable-500}, \Cref{fig:3-4literalGadgets-500}, and \Cref{fig:3-4clause-500}) were an edge of $G$. Then $G$ is a combinatorial triangulation, and each augmenting edge can be realized by a flip of one of these dotted edges. After adding these dotted lines as edges of $G$, our arguments from the $3 \to 4$ augmentation reduction easily translate.

\begin{figure}[h!]
    \centering
    \begin{subfigure}{0.3\textwidth}
            \centering
            \includegraphics[width=\textwidth]{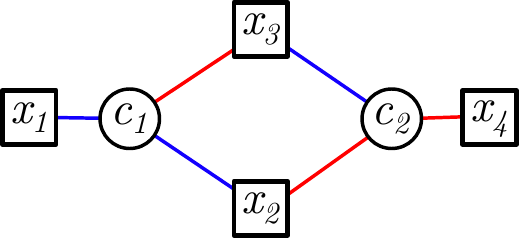}
            \subcaption{}
            \label{fig:reduction-formula}
        \end{subfigure}
        \hfill
    \begin{subfigure}{0.65\textwidth}
        \centering
        \includegraphics[width=1\textwidth]{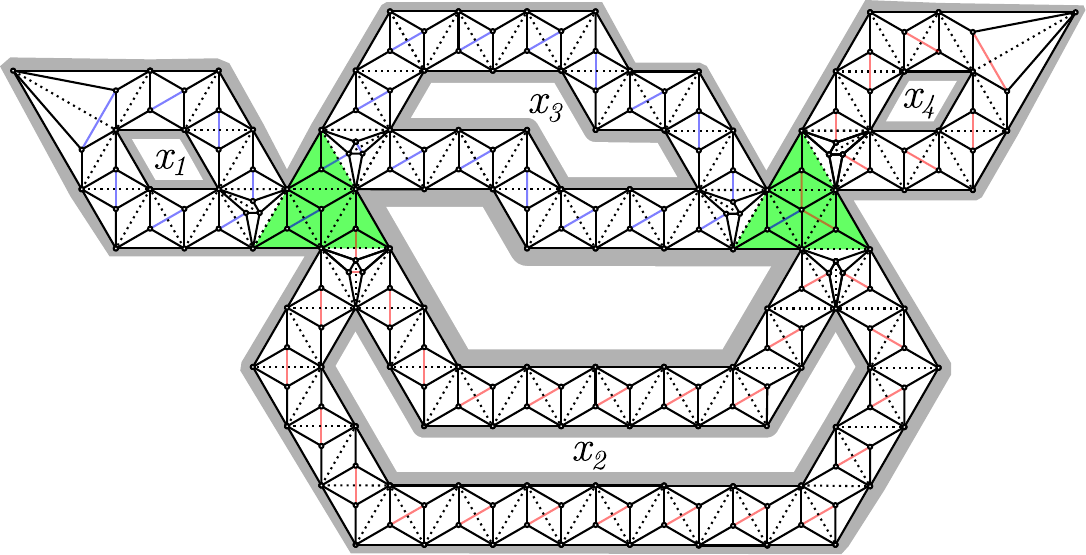}
        \subcaption{}
        \label{fig:ex-reduction-graph}
    \end{subfigure}

    \caption{(a) A \planarsat formula $\varphi = (x_1 \lor x_2 \lor \lnot x_3) \land (\lnot x_2 \lor x_3 \lor \lnot x_4)$ (b) The graph $G$ constructed from $\varphi$ during the reduction, $x_1$ and $x_3$ are assigned true, and $x_2$ and $x_4$ are set to false. Clause gadgets are highlighted in green for legibility. }
    \label{fig:ExampleReductionM}
\end{figure}

\subsection{\texorpdfstring{$4\to 5$}{4 to 5} Connectivity Augmentation is NP-Complete}

We finish our discussion on computational complexity by showing that minimally planar augmentation of a 4 connected plane graph to 5 connectivity is NP-complete. This discussion will be rather brief and nearly identical to our treatment of 3-4 augmentation. Instead of using $W_4$'s as building blocks we use $W_5$'s, (wheel graphs on 5 vertices) with two opposing edges removed, see \Cref{fig:4-5BasicUnit}. Each has four boundary vertices forming a $4$-cut disconnecting the central vertex. With this our variable and clause gadgets immediately follow from our 3-4 augmentation constructions. Our literal gadget is modified a little more, but its function remains the same as in the 3-4 augmentation proof. 

\begin{figure}[ht]
    \centering
    \begin{subfigure}{0.2\textwidth}
        \centering
        \includegraphics[width=.75\textwidth]{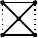}
        \subcaption{}
        \label{fig:4-5BasicUnit}
    \end{subfigure}
    \hfill
    \begin{subfigure}{0.65\textwidth}
        \centering
        \includegraphics[width=\textwidth]{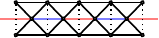}
        \subcaption{}
        \label{fig:4-5Variable}
    \end{subfigure}
    \caption{(a) Modified $W_5$ (b) Variable gadget for $4\to 5$ connectivity augmentation}
    \label{fig:4-5VarialbleGadgets}
\end{figure}

\Cref{fig:4-5VarialbleGadgets} presents our modified $W_5$ and our 4-5 graph augmentation variable gadget. These constructions are fairly obvious transcriptions from our 3-4 graph augmentation constructions. Accordingly, an equivalent transcription of \Cref{lemma:3-4Variable} requires a nearly identical argument, just replacing $W_4$ with $W_5$. 
Due to this symmetry and considerations for the length of this paper we present the following lemma without proof.

\begin{restatable}{lemma}{45Varaible} \label{lemma:4-5Variable}
    
    A variable gadget as presented in \Cref{fig:4-5Variable} constructed from $2k$ $W_5$'s can be augmented to 5-connectivity with $k$ edges.
\end{restatable}

We now present our clause gadget for 4-5 graph augmentation. Just as we did with the variable gadget, we construct this gadget by taking our 3-4 clause gadget and replacing each $W_4$ with our $W_5$'s. This gadget is presented in \Cref{fig:4-5Clause}, and again just like we did with our 4-5 augmentation variable gadget we present a lemma analogous to \Cref{lemma:3-4Clause} without proof, as the arguments from the original can be repurposed with minimal effort.  

\begin{restatable}{lemma}{45clause} \label{lemma:4-5clause}
    The clause gadget presented in \Cref{fig:4-5Clause} can be augmented with precisely three edges if and only if an adjacent literal gadget is true. Otherwise four edges are necessary.
\end{restatable}

\begin{figure}[ht]
    \centering
    \begin{subfigure}{.3\textwidth}
        \includegraphics[width=\textwidth]{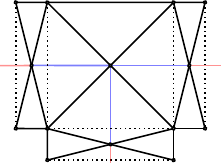}
        \subcaption{}
    \end{subfigure}
   \hfill 
    \begin{subfigure}{.6\textwidth}
        \includegraphics[width=0.85\textwidth]{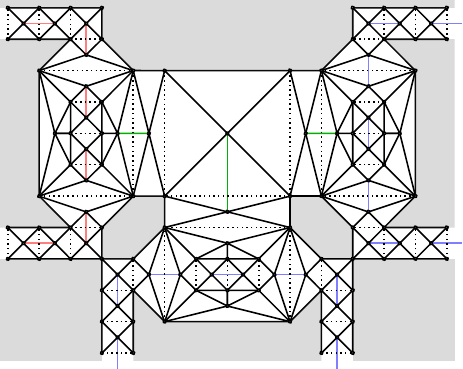}
        \subcaption{}
    \end{subfigure}
    \caption{A clause gadget for $4\to 5$ connectivity augmentation}
    \label{fig:4-5Clause}
\end{figure}

Although the variable and clause gadgets were easily translated, this is not entirely the case for our 3-4 literal gadget. Instead of replacing two $W_4$'s in a variable gadget, like in the 3-4 augmentation case. Our 4-5 literal will replace six $W_5$'s of a variable gadget. The general idea of the gadget, however, remains unchanged. A literal gadget is inserted into variable $x_i$ and an edge of the literal is identified with an edge of a clause gadget $c_j$. If the literal gadget represents a positive literal, then if $x_i$ is false a vertex of $c_j$ must be adjacent to a vertex of the literal gadget. If $x_i$ is true then the literal gadget will be fully augmented and $c_j$ will have the freedom to augment its central vertex. If the literal gadget represents a negated literal $(\lnot x_i)$ then the previous is flipped.

\begin{figure}[ht]
    \centering
    \includegraphics[width=0.9\linewidth]{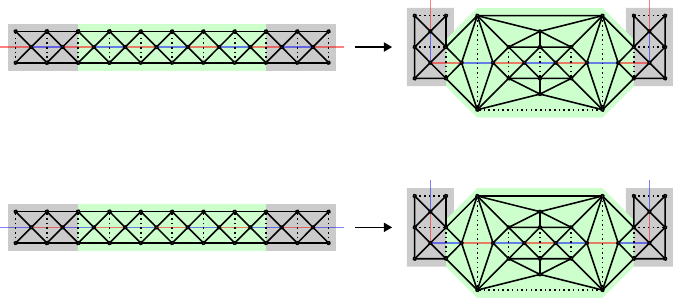}
    \caption{Literal Gadget for 4-5 Graph Augmentation (Top) A positive literal (Bottom) A negated literal }
    \label{fig:4-5Literal}
\end{figure}

We arrive at our main result, that 4-5 augmentation is NP-complete. 
In keeping with the rest of this section, we omit a proof as the proof for \Cref{lemma:3-4NPC} can be used almost word-for-word with some obvious modifications (3 connected should become 4-connected etc\dots). 

\begin{restatable}{lemma}{45NPC}
    \label{lemma:4-5NPC}

    Given a 4-connected planar graph $G=(V,E)$ and an integer $\tau$, it is NP-hard to decide whether there is a set $F$ of edges with $|F|\le \tau$ such that $G'=(V,E\cup F)$ is a 5-connected planar graph.
\end{restatable}

\begin{restatable}{corollary}{45Topo}
    
    \label{cor:45-topo}
    \cref{lemma:4-5NPC} still holds if $G$ is a topological (or geometric) plane graph and if $G'$ is required to be a compatible topological (resp., geometric) plane or 1-plane graph.
\end{restatable}

Just as with \Cref{cor:34-topo}, settings for the topological geometric setting follow easily by just embedding each gadget as depicted. Further, gadgets presented in this section have some dashed lines. Replacing these with edges in $G$ will force any solution to be 1-planar.

\input{ptas}
\input{traingulations}
\input{pslg}

\bibliography{references}

\end{document}

%% file: intro.tex
\section{Introduction}
\label{sec:intro}

Connectivity augmentation is a classical problem in combinatorial optimization~\cite{FJ15,Jordan95}: Given a $c$-connected graph $G=(V,E)$ and a positive integer $k$, find a minimum set $F$ of new edges such that the graph $G'=(V,E\cup F)$ is $k$-connected. For undirected and unweighted graphs, Jackson and Jord\'an~\cite{JacksonJ05a} gave a polynomial-time algorithm for any $c$ and $k$ (previous results addressed $c=2,3,4$~\cite{EswaranT76,Hsu00,WatanabeN93}; see also~\cite{Vegh11}). For the weighted version, finding a minimum-cost augmentation is NP-complete~\cite{FredericksonJ81}, APX-hard, and the best approximation ratio is 1.5~\cite{TraubZ22}. In the remainder of this paper, we consider unweighted undirected graphs.

We investigate connectivity augmentation over planar graphs in three different settings, \Cref{fig:Intro} depicts an illustrative example:
In the \emph{abstract graph setting}, $G$ is planar, and $G'$ must also be planar. For this variant, NP-completeness~\cite{KantB91} and a $\frac53$-approximation are known for $k=2$~\cite{FialkoM98,GutwengerMZ09b}. Moreover, there exists an $O(n^3)$-time $\frac54$-approximation algorithm for $c=2$ and $k=3$~\cite{KantB91}. In the \emph{topological graph setting}, $G$ and $G'$ are plane graphs (with a fixed embedding) and $G'$ extends the given embedding: the problem remains NP-complete for $c=0$ and $k=2$, but there is a linear-time algorithm for $c=1$ and $k=2$~\cite{GutwengerMZ09b,GutwengerMZ09a}. 
Finally, in the \emph{geometric graph setting}, both $G$ and $G'$ are planar straight-line graphs (PSLG, for short), and $G'$ extends the given embedding~\cite{AbellanasOHTU08,AkitayaINSTW19, AloupisBCDFM15,HurtadoT13,RutterW12}. The three settings reveal deep structural properties of embedded graphs as the edges of the input graph $G$ form topological or geometric obstructions to possible new edges.  

    \begin{figure}[ht]
    \centering
    \includegraphics[page=1]{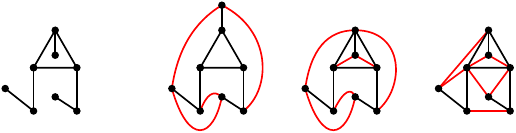}
    \caption{A planar straight-line graph, and its 3-connected minimum augmentation in the abstract, topological, and geometric settings with 4, 6, and 7 additional edges, respectively.
    }
    \label{fig:Intro}
    \end{figure}

However, planarity severely limits the feasible values of $k$: Every planar graph has a vertex of degree at most 5, and hence augmentation is infeasible for any $k>5$. Furthermore, for $n\geq 3$ points in convex position, every PSLG has a vertex of degree at most 2. The connectivity augmentation problem for PSLGs is already challenging for $k=1,2$; and requires additional assumptions about the point configurations for $k=3$; see~\cite{GarciaHHTV09}.
 
In this paper, we also consider beyond-planar variants of the problem: The augmented graph $G'$ need not be planar but ``close'' to planarity. In particular, we study the trade-offs between connectivity and local crossing number, defined as follows. An abstract graph $G$ is \emph{$\ell$-planar} if it admits a drawing in the plane such that each edge has at most $\ell$ crossings. The \emph{local crossing number} of $G$, denoted $\lcr(G)$, is the minimum integer $\ell$ such that $G$ is $\ell$-planar. 

We also consider connectivity augmentation restricted to \emph{topological graphs} and \emph{geometric graphs}. In both cases, $V$ is a set of $n$ points in the plane in general position (i.e., no three points are collinear), the edges are Jordan arcs in topological graphs and straight-line segments in geometric graphs. In both cases, the drawing of $G$ is fixed, and the \emph{local crossing number} is defined as the maximum number of crossings per edge \emph{in that drawing}. 

\subparagraph{Asymptotic Bounds and Tradeoffs.}
We start with an asymptotically tight tradeoff between connectivity and local crossing number (regardless of the number of new edges): Every $k$-connected graph is $O(k^2)$-planar (\Cref{pp:general}). For geometric graphs on $n$ points in convex position, we determine the precise tradeoff (\Cref{pp:circulant}). The asymptotically tight bound can be attained even if we augment a given plane triangulation in the topological setting (\Cref{thm:topological}). The geometric setting turns out to be much more constrained: We show that any PSLG can be augmented to a 3-connected 5-planar geometric graph, and this bound is the best possible (\Cref{thm:1}). However, augmenting a PSLG to a 4-connected geometric graph may already increase the local crossing number to $\Omega(n)$  (\Cref{thm:fan}).

\subparagraph{Algorithmic Results.}
\cref{tab:complexity} presents an overview of the computational complexity of finding minimum solutions; each field has an entry for the abstract, topological, and geometric setting (in this order); a single entry means that the same result holds in all three settings.

\begin{table}[ht]
    \centering
    \begin{tabular}{c|c|c|c|c|c}
         From/To & 1 & 2 & 3 & 4 & 5\\
         \hline
        Tree & --- & P~\cite{Kant96},\cite{GutwengerMZ09b},\cite{AkitayaINSTW19}  & P~\cite{Kant96}, \textbf{P}, ? & ? & ?\\
         \hline
        0 & P & NPC~\cite{KantB91},\cite{GutwengerMZ09a},\cite{Rappaport89} & NPC~\cite{RutterW12} & NPC~\cite{RutterW12} & NPC~\cite{RutterW12}\\
         \hline
        1 & --- & NPC\cite{KantB91}, P\cite{GutwengerMZ09b}, P\cite{AkitayaINSTW19}  & NPC\cite{KantB91},\cite{RutterW12},\cite{AkitayaINSTW19} &  NPC\cite{KantB91},\cite{RutterW12},\cite{AkitayaINSTW19}& NPC\cite{KantB91},\cite{RutterW12},\cite{AkitayaINSTW19}\\
         \hline
        2 & --- & --- & \textbf{NPC} & \textbf{NPC}$(\downarrow)$ & \textbf{NPC}$(\downarrow)$\\
         \hline
        3 & --- & --- & --- & \textbf{NPC}  & \textbf{NPC}$(\downarrow)$\\

         \hline
        4 & --- & --- & --- & --- & \textbf{NPC}
    \end{tabular}
    \caption{Summary of complexity results. 
    Bold results are proved in this paper. Arrows indicate that a result is implied by the one pointed to.}
    \label{tab:complexity}
\end{table}
\subparagraph{Trees.}
Kant~\cite{Kant96} gave an $O(n)$-time algorithm for the 2- and 3-connectivity augmentation of outerplanar graphs with $n$ vertices in the abstract setting, and Gutwenger et al.~\cite{GutwengerMZ09b} gave an $O(n(1+\alpha(n))$-time algorithm for the $1\to 2$ augmentation for connected plane graphs. For a PSLG tree on $n$ vertices, there is an $O(n^4)$-time algorithm for 2-connectivity augmentation~\cite{AkitayaINSTW19}. (Biconnectivity augmentation can also be solved optimally over outerplanar graphs, where both $G$ and $G'$ must be outerplanar~\cite{OlaverriHNT10}.) 
We give an $O(n)$-time algorithm for the minimum augmentation of a plane tree to a 3-connected plane graph in the topological settings (\Cref{thm:tree}).
In the geometric setting, 3-connectivity augmentation of trees remains open.  

\subparagraph{Geometric Triangulations.}  For a straight-line triangulation on $n$ points in convex position, we give a combinatorial characterization of 3- and 4-connected augmentations.
In this setting, the dual graph is a tree.
We use the dual tree for a dynamic programming algorithm to compute a minimum augmentation to a $ k$-connected $\ell$-planar geometric graph for $k\in \{3,4\}$ and constant $\ell\in \mathbb{N}$ in $O(n)$ time (\Cref{thm:dp}).  

\subparagraph{Hardness.} In most settings, however, the $c\to k$ connectivity augmentation problem is NP-complete. We prove the following theorem using several reductions from \planarsat. 
\begin{theorem}
    \label{thm:hardness} Given a $c$-connected planar graph $G$ and an integer $\tau$, deciding whether there is an edge set of cardinality at most $\tau$ that augments $G$ into a planar $k$-connected $G'$ is NP-complete when $c\in\{2,3,4\}$, $k\in\{3,4,5\}$, and $c<k$.
    The problem remains NP-complete if $G$ is a topological or geometric graph and $G'$ is required to be a 1-plane graph extending the embedding of $G$.
\end{theorem}

Our reductions presented in \Cref{sec:hardness} (from \planarsat) work similar to \cite{AkitayaINSTW19} in the sense that the truth assignment of a variable is encoded in one of two possible perfect matchings on a subset of vertices whose degree must be augmented. 
The proof in \cite{AkitayaINSTW19} is specific to the geometric setting, and it is not trivial to generalize it to the other two settings. 
Moreover, $c$-connected gadgets for $c\in\{2,3,4\}$ required delicate new designs to ensure that the desired behavior in the reduction~\cite{AkitayaINSTW19} go through, especially when $G$ is a triangulation.

\subparagraph{Minimum-Degree Augmentation.}
In a $k$-connected graph $G$ on $n>k$ vertices, the minimum vertex degree $\delta(G)$ is at least $k$. All new hardness reductions in \Cref{tab:complexity} hinge on increasing the degree of a set of special vertices from $c$ to $k$, using the minimum number of new edges, while maintaining planarity. 
As a consequence, we also prove that the following \emph{min-degree augmentation} problems are NP-complete: 

\begin{corollary}    
Given a planar graph $G$ with $\delta(G)=c$ and an integer $k$, find a minimum set $F$ of new edges such that $G'=(V,E\cup F)$ is planar and satisfies $\delta(G')\geq k$. 
\end{corollary}

\subparagraph{Flip Distance to a 4-Connected Triangulation.} A \emph{flip} (a.k.a. \emph{edge exchange}) in a combinatorial triangulation (i.e., edge-maximal planar graph) is an elementary operation that removes one edge and inserts a new edge, producing another triangulation. The \emph{flip diameter} on $n$ vertices is the minimum number of flips that can transform any combinatorial triangulation into any other. Current bounds on the flip diameter use a canonical form, which is a triangulation with a Hamiltonian cycle. Since 4-connected triangulations are Hamiltonian~\cite{Tutte56} (see also~\cite{KawarabayashiO15,OzekiZ18,ThomasY94} for extensions to beyond-planar graphs),
the flip distance to a Hamiltonian or a 4-connected triangulation was studied~\cite{BoseV11,CardinalHKTW18}. Bose et al.~\cite{bose2014making} posed the following problems: Given a combinatorial triangulation $T$, find the minimum number of flips that can take it to a 4-connected (resp., Hamiltonian) triangulation. Our hardness reduction for the $3\to 4$ connectivity augmentation problem can be modified to show that finding the flip distance to 4-connectedness is also NP-complete (\Cref{cor:flip-hardness}).

A \emph{simultaneous flip} in a combinatorial triangulation is an operation that performs one or more edge flips simultaneously (each of which is a valid flip, and they can be performed independently)~\cite{GaltierHNPU03}. Bose et al.~\cite{BoseCGMW07} showed that every triangulation on $n\geq 6$ vertices can be transformed into a 4-connected triangulation with a single simultaneous flip; the number of flipped edges is at most $\floor{(2n-7)/3}$~\cite{CardinalHKTW18}. We show that it is NP-complete to find a minimum simultaneous flip to 4-connectivity (\Cref{cor:flip-hardness}). 
By inserting all simultaneously flipped edges (without removing any edges), we obtain a 4-connected 1-planar graph. In fact, finding a minimum augmentation to a 4-connected 1-planar graph in a triangulation is equivalent to finding a minimum simultaneous flip to 4-connectivity~\cite[Section~3]{BoseCGMW07}. 

\begin{corollary}
    \label{cor:flip-hardness}
    Given a combinatorial triangulation $T$ and an integer $\tau$, it is NP-complete to decide whether there exists a sequence of at most $\tau$ flips that transforms $T$ into a 4-connected triangulation. It is also NP-complete to decide whether there is a single simultaneous flip of cardinality at most $\tau$ that transforms $T$ into a 4-connected triangulation.
\end{corollary}

In \Cref{sec:ptas}, we give an EPTAS for the first of the two problems  (\Cref{thm:PTAS}). We also give some intuition for why our techniques do not directly apply to the second problem.

\subparagraph{Further Related Work.}
While we aim for increasing the \emph{vertex-connectivity} of a given graph, analogous problems can be considered for \emph{edge-connectivity}~\cite{Al-JubehIRSTV11,CenLP22,JohansenRT25,NagamochiE98,Toth12,WatanabeN87}. Some of our techniques might extend to edge-connectivity, however, such adaptations are not immediate and are beyond the scope of this paper.  

We have measured the distance to planarity by the local crossing number. However, there are many other notions of beyond-planarity. For example, Garc\'{\i}a et al.~\cite{GarciaHHTV09} studied the connectivity augmentation problem for \emph{geometric biplane graphs}, which admit a straight-line drawing in the plane and a 2-coloring of the edges such that no two edges of the same color cross  (a.k.a.\ graphs with geometric thickness 2). They showed that every sufficiently large point set in the plane in general position admits a 5-connected geometric biplane graph, and the connectivity bound 5 cannot be improved. Furthermore, every PSLG (other than a wheel or a fan) can be augmented into a 4-connected geometric biplane graph, and the connectivity bound 4 cannot be improved. 

%% file: warmup.tex
\section{Tradeoff between Connectivity and Local Crossing Number}
\label{sec:Observations}

As a warmup exercise, we start with the following basic question, which corresponds to augmenting the empty graph to a $k$-connected graph with local crossing number at most $\ell$. 
\begin{question}\label{Q1}
Given $\ell\in \mathbb{N}$, what is the maximum connectivity of an $\ell$-planar graph?
\end{question}
\begin{proposition}\label{pp:general}
    For every $\ell\in \mathbb{N}$, the connectivity of an $\ell$-planar graph is $O(\sqrt{\ell})$. This bound is the best possible and can be attained by a geometric graph on any point set in general position. Every set of $n\geq k+1$ points in the plane (in general position) admits a $k$-connected $O(k^2)$-planar geometric graph; and this bound is the best possible. 
\end{proposition}
\begin{proof}
    Let $G=(V,E)$ be a $k$-connected $\ell$-planar graph on $n=|V|$ vertices. Then $\deg(v)\geq k$ for every $v\in V$, and so $|E|\ge kn/2$. 
    Assume w.l.o.g.\ that $k\geq 8$, hence $|E|\geq 4n$. By the Crossing Lemma~\cite{ACNS82,Bungener024a,L83}, the total number of crossings in any drawing of $G$ is $\mathrm{cr}(G)=\Omega(|E|^3/n^2) = \Omega(k^3n)$. 
    By the pigeonhole principle,  there is an edge with $\Omega(\mathrm{cr}(G)/|E|)=\Omega(|E|^2/n^2)=\Omega(k^2)$ crossings; hence $\ell=\Omega(k^2)$ and $k=O(\sqrt{\ell})$.

    For a matching lower bound, let $V$ be the set of $n\geq k+1$ points in the plane in general position. Assume w.l.o.g.\ that they have distinct $x$-coordinates. Sort the points by $x$-coordinate and connect every point to its $k$ neighbors on the left and right. The first $k+1$ points induce a complete graph, which is $k$-connected. By induction, the first $i$ points also induce a $k$-connected graph for $k<i\leq n$. Any edge can only cross other edges with overlapping $x$-projections, so every edge has $O(k^2)$ crossings.  
\end{proof}
Determining the precise dependency between $k$ and $\ell$ remains largely open. One exception is the case $k=1$: The maximum connectivity of a 1-planar graph is 7. Every 1-planar graph on $n\geq 3$ vertices has at most $4n-8$ edges~\cite{PachT97}, so the minimum degree is at most $7$, hence its connectivity is at most 7. This degree bound is tight: Biedl~\cite{Biedl21} constructed 1-planar graphs with minimum degree 7, which happen to be 7-connected. However, the connectivity of \emph{optimal} 1-planar graphs (which have precisely $4n-8$ edges) is either 4 or 6~\cite{FujisawaSS18,Suzuki10}; there are edge-maximal 1-planar graphs with $\frac{8}{3}n-O(1)$ edges, where the connectivity is only 2~\cite{Hudak12}. 

\subparagraph{Geometric graphs.} 
We can determine the exact tradeoff between connectivity and local crossing number for geometric graphs when the vertices are in convex position.
For integers $k,n\in \mathbb{N}$, $k<n/2$, the \emph{$k$-circulant graph}
$G=(V,E)$ is defined on the vertex set $V=\{v_0,\ldots , v_{n-1}\}$, where $v_iv_j\in E$ if and only if the cyclic distance between $i$ and $j$ is at most $k$, that is, $\min\{|i-j|,n-|i-j|\}\leq k$; see \Cref{fig:circulant}.

\begin{figure}[htbp]
  \centering
    \includegraphics{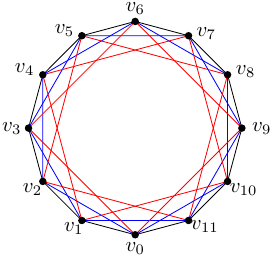}
    \caption{The $3$-circulant graph on $n=12$ vertices.}
    \label{fig:circulant}
\end{figure}

\begin{restatable}{proposition}{ppcirculant}\label{pp:circulant} 
        For every $k\in \mathbb{N}$ and for $n\geq 2k$ points in convex position in the plane, the $k$-circulant graph is $2k$-connected and is $(k^2-k)$-planar. This is the best possible: If a geometric graph on $n$ points in convex position has minimum degree at least $2k$, then its local crossing number is at least $k^2-k$.
\end{restatable}
\begin{proof}
Let $k\in \mathbb{N}$, and let $G=(V,E)$ be the $k$-circulant graph on $n$ vertices, that is, the vertices are on a circle, and each vertex is adjacent to $k$ neighbors in both cw and ccw direction along the circle. Elementary counting shows that every edge crosses at most $k^2-k$ other edges. We show that $G$ is $2k$-connected. We need to show that $G-D$ is connected for any $D\subset V$ with $|D|=2k-1$. Let $s,t\in V\setminus D$. Then $s$ and $t$ decompose the circle into two arcs. One of them contains at most $k-1$ vertices in $D$. We construct an $st$-path $\pi_{st}$ by following this arc, and skip all vertices in $D$. Since we skip at most $k-1$ consecutive vertices, all edges of $\pi_{st}$ are present in $G-D$, and so $G-D$ is connected.

Now let $G=(V,E)$ be a geometric graph on $n\geq 2k$ points in convex position, with minimum degree at least $2k$. We define the \emph{length} of an edge $uv$ as the minimum number of edges in a path on the boundary of the convex hull. Then $|E|\geq kn$, and elementary counting shows that $G$ has at least one edge of length at least $k$. Let $e$ be a shortest edge among all edges of length at least $k$. We claim that if $e$ has length $\ell\geq k$, then it has at least $k^2-3k+2\ell\geq 2k^2-k$ crossings. Assume w.l.o.g.\ that $V=\{v_0,\ldots v_{n-1})$ in cyclic order along the convex hull of $V$, and $e=(0,\ell)$. Every edge that crosses $e$ has on endpoint in $\{v_1,\ldots , v_{\ell-1}\}$ and one in $\{v_{\ell+1},\ldots , v_{n-1}\}$. Then
$\sum_{i=1}^{\ell-1} \deg(v_i)\geq 2k(\ell-1)$. We subtract the number of vertex-edge pairs $(i,e')$ such that $1\leq i\leq \ell-1\}$ and $e'$ does not cross $e$. Every such edge $e'$ has length less than $k$. By symmetry, it is enough to count pairs $(i,e')$ where $e'$ is a ``left'' edge, that is, $e'=\{i,j\}$ and $j<i$. The number of such edges for $i=1,2,\ldots , \ell-1$ is $1,2,3,\ldots, k-1,k-1,\ldots , k-1$, which sums to $\binom{k}{2}+(\ell-k)(k-1)$. So edge $e$ crosses at least 
$    2k(\ell-1) - 2\left(\binom{k}{2}+(\ell-k)(k-1)\right) = k^2-3k+2\ell\geq k^2-k$ edges. This is a lower bound for the local crossing number of $G$.
\end{proof}

%% file: topological.tex
\section{Augmentation for Topological Graphs}
\label{sec:Top}

Every edge-maximal plane graph is 3-connected, but it may have separating triangles. For augmentation to 4-connectivity, we need to go beyond planarity. 

\begin{proposition}
   Every plane graph $G$ on $n\geq 6$ vertices can be augmented to a 4-connected 1-planar topological graph. If $G$ is a triangulation, then $\floor{(2n-7)/3}$ new edges suffice.
\end{proposition}
\begin{proof}
Let $G=(V,E)$ be a plane graph on $n\geq 6$ vertices. We may assume that $G$ is an edge-maximal plane graph (i.e., a triangulation). Cardinal et al.~\cite[Theorem~3]{CardinalHKTW18} proved that $G$ can be transformed into a 4-connected edge-maximal planar graph $G'=(V,E')$ using a \emph{simultaneous flip} of at most $\floor{(2n-7)/3}$ edges. Now $G=(V,E\cup E')$ is 4-connected (since it contains $G'$) and 1-planar (an edge has a crossing iff it participates in a flip). 
\end{proof}

For plane graph augmentation, we prove a tight tradeoff (matching~\Cref{pp:general}).

\begin{theorem}\label{thm:topological}
    Every plane graph can be augmented to a $k$-connected $O(k^2)$-planar topological graph, and this bound is the best possible.
\end{theorem}
\begin{proof} 
Without loss of generality assume we are given a plane triangulation $G=(V,E)$ and a $k\in \mathbb{N}$. The dual graph $\mathcal{D}$ of $G$ is a planar graph of maximum degree 3. We can partition  $\mathcal{D}$ into a collection $\mathcal{C}=\{C_1,\ldots , C_t\}$ of connected subgraphs, each containing at least $2k-1$ and at most $6k-2$ nodes as follows. Consider an arbitrary spanning tree $\mathcal{T}$ of  $\mathcal{D}$, and recursively partition $\mathcal{T}$ into subtrees by deleting an edge incident to a centroid node, until any further partition would produce a subtree with less of than $2k-1$ nodes (see \Cref{fig:clustering}).

\begin{figure}[htbp]
\hfill
     \begin{subfigure}[b]{0.32\textwidth}
          \centering
         \includegraphics[width=45mm]{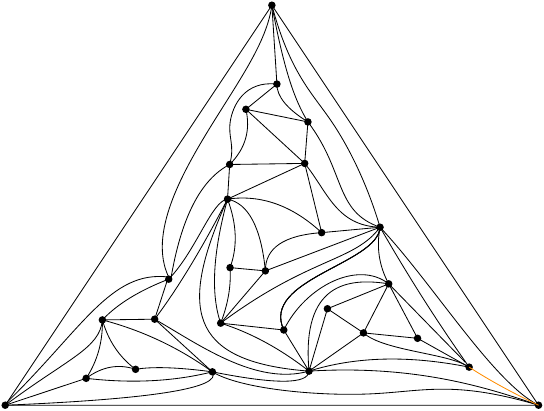}
    \subcaption{}
    \label{fig:traingulation}
     \end{subfigure}
  \centering
     \begin{subfigure}[b]{0.32\textwidth}
          \centering
   \includegraphics[width=45mm]{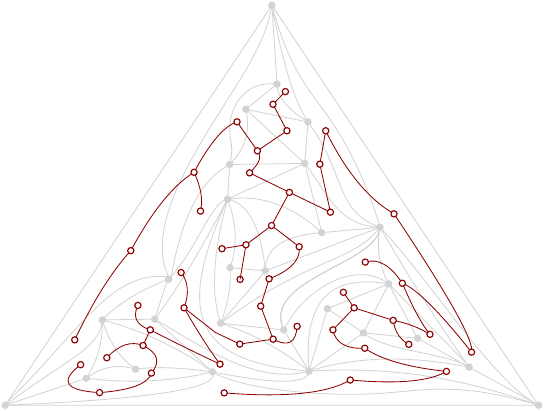}
    \subcaption{}
    \label{fig:dual}
     \end{subfigure}
     \hfill
     \begin{subfigure}[b]{0.32\textwidth}
          \centering
   \includegraphics[width=45mm]{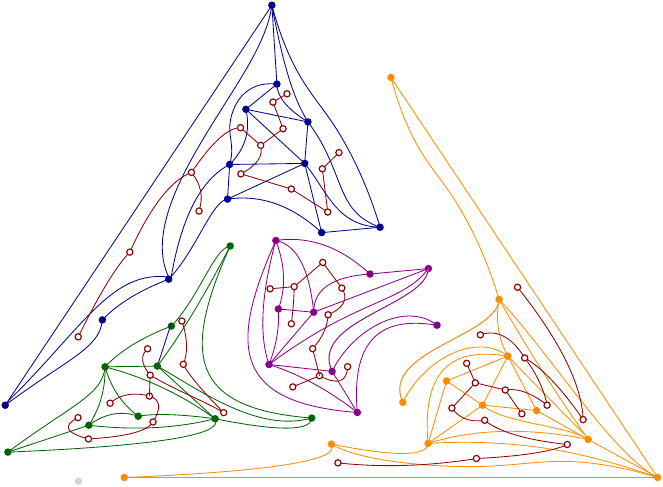}
      \subcaption{}
    \label{fig:clusters}
     \end{subfigure}
       \caption{A triangulation on $n=27$ vertices (a); a spanning tree $\mathcal{T}$ of the dual graph $\mathcal{D}$ (b); four clusters obtained for $k=5$ (c). }
       \label{fig:clustering} 
       \end{figure}

For every $i=1\ldots , t$, let $R_i$ denote the region formed by the union of triangles in $C_i$; and let $V_i\subset V$ denote the set of all vertices of the triangles in $C_i$. By Euler's formula, applied for the triangles in $C_i$, we have $k\leq |V(C_i)|\leq 6k$.
We call the vertex sets $V_1,\ldots , V_t$ \emph{clusters},
We also define a \emph{cluster graph} $\mathcal{G}$, where the nodes correspond to the $t$ clusters, and nodes $V_1$ and $V_j$ are adjacent in $\mathcal{G}$ if $C_i$ and $C_j$ contain two triangles that share an edge. 
Note that a vertex $v\in V$ may be part of arbitrarily many clusters (e.g., if $v$ is a high-degree vertex), yet each cluster is adjacent to $O(k)$ other clusters in $\mathcal{G}$.

We can augment $G$ as follows. (1) Augment each cluster $V_i$ to a clique, drawing all new edges along shortest paths in the region $R_i$. (2) For each pair of adjacent clusters, $V_i$ and $V_j$, add a matching of size $k-|V_i\cap V_j|$ between $V_i\setminus V_j$ and $V_j\setminus V_i$, drawing the new edges along shortest paths in the region $R_i\cup R_j$.  Let $G'$ denote the augmented graph.

First, we show that $G'$ is $k$-connected. Specifically, we show by induction on $i\geq 1$ that if $i$ clusters induce a connected subgraph of the cluster graph, then the union of the clusters induce a $k$-connected graph in $G'$. The claim is clear in the base case $i=1$. For the induction step, we use Menger's Theorem. Assume that the claim holds for $i-1$ clusters. The union of $i$ clusters is composed of two parts, $A$ and $B$, where $A$ is the union of $i-1$ clusters, including cluster $V_j$, and $B=V_i$ is the $i$-th clusters. Both $G'[A]$ and $G'[B]$ are $k$-connected by the induction hypothesis, they share $|V_i\cap V_j|\geq 2$ vertices, and are connected by a matching of size $k-|V_i\cap V_j|$ (which is disjoint from the shared vertices). Let $a_1=b_1, \ldots ,a_s=b_s$ be the shared vertices, and let $a_t b_t$, $t=s+1,\ldots k$, be the matching with $a_i\in A\setminus B$ and $b_i\in B\setminus A$. Let $a\in A$ and $b\in B$ be arbitrary vertices. By Menger's Theorem, there are $k$ internally disjoint paths from $a$ to $a_1,\ldots , a_k$ in $G'[A]$, hance $k$ internally disjoint path from $a$ to $b_1,\ldots ,b_k$ in $G'[A\cup B]$. Since $G'[B]$ is a clique, we obtain $k$ internally disjoint paths between $a$ and $b$ in $G'[A\cup B]$. By Menger's Theorem, $G'[A\cup B]$ is $k$-connected. 

Second, we claim that $G'$ is $O(k^2)$-planar. Consider an edge $e$ in a cluster $C$. In the complete graph of $C$, edge $e$ is crossed by $O(k^2)$ other edges. Cluster $C$ is adjacent to $O(k)$ other clusters, each of which is connected to $C$ by a matching of size $O(k)$. Edge $e$ crosses $O(k^2)$ edges in these matchings. Therefore, any edge in a cluster crosses $O(k^2)$ other edges in $G'$. 
Now consider an edge $f$ of a matching between two adjacent clusters $A$ and $B$. Then $f$ can cross any of the $2\cdot O(k^2)$ edges in the two cliques induced by clusters $A$ and $B$. It can also cross the $2\cdot O(k^2)$ edges in the $O(k)$ matchings between $A$ (resp., $B$) and their adjacent clusters. 
Overall, any edge in $G'$ has $O(k^2)$ crossings.
\end{proof}

%% file: trees.tex
\section{\texorpdfstring{$1\to 3$}{1 to 3} Connectivity Augmentation for Plane Trees}
\label{sec:trees}

In this section we present a linear-time algorithm for the following problem:
Given an $n$-vertex plane tree $T$ (i.e., with a fixed embedding), compute a $3$-connected plane graph $G$ with the minimum number of edges such that $T$ is a spanning subgraph of $G$.
That is, we solve the $3$-connectivity augmentation of plane trees in the topological setting in linear time.
We remark that Dhanalakshmi, Sadagopan, and Manogna~\cite{DhanalakshmiSM16} solve the problem for a tree $T$, where $G$ is not necessarily planar.
However, we could not fully verify the correctness of their algorithm. In their last step, the algorithm seems to introduce a matching on a long path of degree-$2$ vertices in $T$, which does not result in a $3$-connected graph.

Let us start with a simple lower bound (which also appears in \cite{DhanalakshmiSM16}) on the number of inserted edges. 
For a graph $G = (V,E)$ and integer $i \geq 0$, let $n_i(G) = |\{v \in V \colon \deg(v) = i\}|$ denote the number of vertices in $G$ of degree~$i$.

\begin{observation}
    \label{obs:trees-LB-500}
    Let $T$ be a spanning tree in a $3$-connected graph $G$.
    Then $G$ has minimum degree $\delta(G) \geq 3$.
    In particular, $|E(G)| - |E(T)| \geq \ceil*{\frac{2n_1(T) + n_2(T)}{2}}$.
\end{observation}

We present an algorithm that meets the lower bound of \Cref{obs:trees-LB-500}.

\begin{restatable}{theorem}{tree}\label{thm:tree}
    Let $T$ be an $n$-vertex plane tree with $n\geq 4$ and $n_i(T)$ vertices of degree $i$.
    Then there exists a plane $n$-vertex graph $G$ with $T \subseteq G$ and $|E(G)| - |E(T)| = \ceil*{\frac{2n_1(T) +n_2(T)}{2}}$.

    Moreover, such a $G$ can be computed in $O(n)$ time.
\end{restatable}
\begin{proof}
    We present our construction of $G$ as an inductive proof on $n$, which can be easily turned into a linear-time recursive algorithm. For simplicity, we call the edges of $E(G)\setminus E(T)$ the \emph{new} edges and illustrate them in red.
    For the induction base of $n=4$, we set $G=K_4$.
    So let us henceforth assume that $n \geq 5$.
    We use a stronger induction hypothesis, where we require that $n_2(T)$ is even; we later show how to handle the odd case.
    This assumption ensures that each leaf of $T$ is incident to exactly two new edges, each vertex of degree $2$ to exactly one new edge, and all other vertices to no 
    new edges.
    We distinguish four cases. 

      \textbf{Case 1:}
            A vertex $v$ has two consecutive neighbors $a,b$ of $\deg(a) = \deg(b) = 1$.
            
            Say $a$ comes clockwise after $b$ at $v$.
            We remove $a$ and $b$ from $T$.
            If $v$ is not a leaf after this removal, then we add another leaf $w$ at $v$; otherwise $v$ takes the role of $w$.
            We denote the new tree by $T'$ and note that $n_2(T') = n_2(T)$, i.e., the number of degree-$2$ vertices remains the same.
            In order to apply the induction, we must also ensure that $T'$ has at least $4$ vertices if no new leaf was added, i.e., $T$ has at least $6$ vertices.
            This is the case as otherwise $T$ corresponds the tree $T_5^*$ depicted in \Cref{fig:treeBase}, for which $n_2(T_5^*)$ is odd. 
            In the obtained solution $G'$ with $T' \subseteq G'$, vertex $w$ is incident to two edges in $F' = E(G') - E(T')$, say $wx$ comes clockwise after $wy$ at $w$.
            Then we remove $wy$ and insert $ab$ and $by$.

    \begin{figure}[ht]
    \centering
    \includegraphics[page=2]{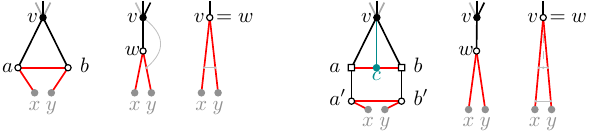}
    \caption{Illustration of Case 1 and 2 in the proof of \Cref{thm:tree}.}
    \label{fig:tree}
    \end{figure}
  
    We define a \emph{leg} as a vertex $v$ in $T$ of $\deg(v)=2$ with neighbor $a$ of $\deg(a) = 1$.
    
  \textbf{Case 2:}
            A vertex $v$ has two neighbors $a,b$ that are legs which are either consecutive or have a leaf $c$ between them. 
            
            We may assume that $a$ comes clockwise after $b$, and if it exists, $c$ comes after $b$ and before $a$ at $v$.
            Then we remove $a,b$, their adjacent leaves $a',b'$ and $c$.
            If $v$ is not a leaf after this removal, then we add another leaf $w$ at $v$; this ensures that $v$ does not have degree $2$.
            As before, we aim to recurse on the resulting tree $T'$.
            Note that $n_2(T') = n_2(T) - 2$ because $a$ and $b$ are removed and $v$ has degree $1$ or at least $3$ in $T'$.
            We also ensure that $T'$ has at least $4$ vertices.
            In the case that we introduce a new leaf $w$, this is due to the fact that $v$ has at least two more neighbors.
            In the other case, we must exclude the case that $T'$ is a path on three vertices.
            However, this would imply that $T$ is either the tree $T_7^*$ or $T^*_8$, illustrated in \Cref{fig:treeBase}, both of which have an odd number of degree-$2$ vertices.
            Thus, we may apply induction on $T'$.
            
             \begin{figure}[ht]
                \centering
                \includegraphics[page=5]{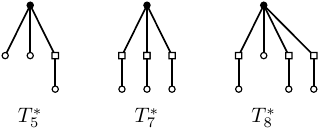}
                \caption{Illustration for the (non-existing) base cases in the proof of \Cref{thm:tree}.}
                \label{fig:treeBase}
             \end{figure} 
             
            In the obtained solution $G'$, vertex $w$ has two incident edges in $F' = E(G') - E(T')$, say $wx$ comes clockwise after $wy$ at $v$.
            Then we remove $wx,wy$ and insert $a'x$, $a'b'$, and $b'y$.
            If $c$ exists in $T$, we insert $ac$ and $cb$, otherwise we insert only $ab$.

         \textbf{Case 3:}
            A vertex $v$ with $\deg(v)=2$ has a neighbor $b$ that is a leg adjacent to leaf $a$.
            
            In this case we delete $a$ and $b$ and recurse on $T' = T-\{a,b\}$; note that $n_2(T') = n_2(T) - 2$.
            Moreover, if $T'$ had only $3$ vertices, then $T$ would correspond to a path on $5$ vertices for which $n_2(T)$ would be odd.
            Hence, we may apply the induction. 
            In the obtained solution $G'$, vertex $v$ has two incident edges in $F' = E(G') - E(T')$, say $vx$ and $vy$.
            Then we remove $vx,vy$ and insert $va$, $ax$, and $by$.
    
    \begin{figure}[ht]
    \centering
    \includegraphics[page=4]{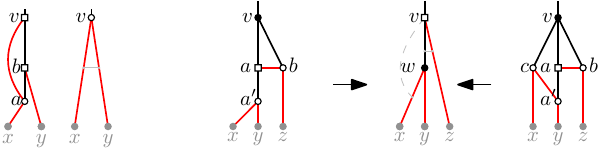}
    \caption{Illustration of Case 3 and 4 in the proof of \Cref{thm:tree}.}
    \label{fig:tree2}
    \end{figure} 
    
        \textbf{Case 4:}
        A vertex $v$ (with $\deg(v)=3$ or $4$) has a neighbor $p$ with $\deg(p)\geq 3$, one neighbor $a$ which is a leg and otherwise only leaves.

        By Case 1, we may assume that no two leaves are consecutive.
        Hence, either $v$ has exactly one leaf $b$ or two leaves $b,c$ which are separated by $p$ and $a$, see also \Cref{fig:tree2}.
        In this case we delete $a,b,c$ and the leaf $a'$ of $a$, and insert a new leaf $w$ at $v$.
        We recurse on the resulting tree $T'$. 
        Note that $n_2(T') = n_2(T)$ as we deleted $a$ and $v$ has degree $2$ in $T'$.
        Moreover, due to the fact that $p$ has $\deg(p)=3$, $T'$ has at least $4$ vertices.
        In the solution $G'$ obtained by induction, vertex $v$ is incident to one edge, say $vz$, and vertex $w$ to two edges, say $wx$ and $wy$, in $F' = E(G') - E(T')$.
        Then we remove $wx,wy,wz$ and insert $a'y$, $ab$, and $bz$.
        If $c$ does not exist, we insert $a'x$ and otherwise $a'c,cx$.

    This completes the construction of $G$ with $T \subseteq G$.
    To show that the case distinction is exhaustive, root the tree $T$ at an arbitrary vertex and consider a lowest vertex $v$ that is neither a leaf nor a leg.
    All but one neighbor $p$ of $v$, are either leaves or legs.
    If $v$ has $\deg(v)= 2$, then Case 3 applies.
    Otherwise, apart from $p$, $v$ has two consecutive leaves (Case 1), two legs which are consecutive or have a leaf between them (Case 2), or one leg and otherwise only leaves (Case 4).

    It remains to argue that the resulting graph $G$ is $3$-connected.
    This is clear in the base case $G = K_4$.
    In each remaining case, it suffices to observe that $G$ is obtained from a $3$-connected graph $G'$ by a local modification that is a sequence of so-called \emph{BG-operations}; these operations were introduced by Barnette and Gr\"unbaum~\cite{BG69} and preserve $3$-connectivity.
    In a \emph{$(1,2)$-operation} we subdivide an edge $xy$ by a new vertex $w$ and insert an edge $vw$ with $v \neq x,y$.
    In a \emph{$(2,3)$-operation} we subdivide edge $xy$ by vertex $a$ and edge $vw$ by vertex $b$ and insert the edge $ab$. \Cref{fig:tree,fig:tree2} illustrate how up to three of these operations suffice.

    Lastly, we show how to handle the case that $n_2(T)$ is odd.
    We choose an arbitrary vertex $v$ of degree $2$ and contract one of its incident edges to obtain a tree $T'$.
    Using our above machinery, we obtain a $3$-connected graph $G'$ with $|E(G')| - |E(T')| = \ceil*{\frac{2n_1(T) +n_2(T)-1}{2}}=\ceil*{\frac{2n_1(T) +n_2(T)}{2}}-1$.
    Now we reintroduce $v$ by a $(1,2)$-operation on its edge, connecting $v$ to an arbitrary non-neighbor in an incident face.
    As this preserves $3$-connectivity, this completes the proof.
\end{proof}

%% file: ptas.tex
\section{EPTAS for Flip Distance of Triangulations to 4-Connectivity}
\label{sec:ptas}

In this section, we provide an EPTAS for making plane triangulations $4$-connected with a sequence of edge flips.
\begin{restatable}{theorem}{PTAS}\label{thm:PTAS}
    For every $\epsilon>0$, there is an $n\cdot 2^{O(1/\epsilon)}$-time algorithm that, for a given triangulation $G$ on $n\geq 6$ vertices, returns a sequence of edge flips that turns $G$ into a 4-connected triangulation, and the length of the sequence is at most $1+\epsilon$ times larger than necessary.  
\end{restatable}

We first reduce the problem to a certain hitting set problem as suggested by Bose \textit{et al.}~\cite{bose2014making}.
The reduction relies on the following.
\begin{lemma}[Mori \textit{et al.} \cite{MoriNO03}]\label{lemma:no_new_separating}
    Let $G$ be a plane triangulation on $n\geq 6$ vertices, $T$ a separating triangle of $G$ and $e$ an edge of $T$. Then flipping $e$ destroys the separating triangle $T$ and does not create any new separating triangles in $G$, provided $e$ is incident to multiple separating triangles or none of the three edges of $T$ are incident to multiple separating triangles.
\end{lemma}

This allows one to easily reduce the problem of finding an optimal sequence of flips to a certain hitting set problem on the edges of the triangulation.

\begin{restatable}{lemma}{PTASreduction}\label{lemma:ptas-reduction}
    Let $G$ be a plane triangulation on $n\geq 6$ vertices. Let $E'$ be a set of edges of $G$ such that every separating triangle in $G$ is incident to at least one edge in $E'$, and let $\tau$ be the minimum size of any such set.

    Making $G$ $4$-connected requires at least $\tau$ edge flips. Moreover, given $E'$, one can compute a sequence of edge flips of length at most $|E'|$ making $G$ $4$-connected in $O(n)$ time.
\end{restatable}
 
\begin{proof}
    If $F$ is a sequence of edge flips making $G$ $4$-connected, then every separating triangle in $G$ must be incident to at least one edge in $E'$ (otherwise the missing separating triangle remains in $G$ after all flips have been performed, and the graph is not $4$-connected). Thus, $F$ has length at least $\tau$.

    Now consider the edges of $E'$ in some arbitrary order $e_1,e_2,\ldots,e_{|E'|}$. 
    Define a sequence of edge flips as follows. Let $1\leq i \leq |E'|$, and assume the first $i-1$ edges of $E'$ have already been processed. Then for the $e_i$, do the following:
    \begin{itemize}
        \item if $e_i$ is not present in $G$ (because it has already been flipped previously) or does not bound any separating triangles in $G$, don't flip any edge;
        \item otherwise, if $e_i$ is incident to multiple separating triangles in $G$, flip $e_i$;
        \item otherwise, if none of the $3$ edges of the unique separating triangle $T$ incident to $e$ are incident to any other separating triangle, flip $e_i$;
        \item otherwise, flip an edge of $T$ incident to at least $2$ separating triangles.
    \end{itemize}

    By the definition of this sequence and of $E$, it is clear that for any separating triangle in $G$, at least one of its $3$ edges gets flipped. Moreover, every edge that is flipped obeys the conditions of Lemma \ref{lemma:no_new_separating} (at the time it is flipped), and thus, no new separating triangle is introduced and the resulting graph is $4$-connected.

    Let us describe how to carry out the procedure in linear time. For each edge, we will maintain a set of pointers to the separating triangles it bounds in a doubly-linked list. Each separating triangle also has corresponding back-pointers, so that from a separating triangle we can access the three pointers pointing to it in constant time.
    We can initialize this structure by computing the set of separating triangles incident to each edge in $O(n)$ total time (this can be done for example by computing a $4$-block decomposition of the triangulation in linear time \cite{Kant97}).

    For a given edge, we can then test in constant time if it bounds $0$, $1$ or at least two separating triangles. In the case it bounds exactly one separating triangle $T$, we can test in constant time whether this is also the case for the two other edges bounding $T$.
    Each time we flip an edge $e$, we go through the set of separating triangles it used to bound by traversing its associated doubly-linked list and following the pointers it stores. For each such separating triangle, we delete it from the doubly-linked lists of the three edges which (used to) bound it by following the back-pointers. This costs constant time per deleted separating triangle. 
    
    Because no new separating triangle is created at any point, every separating triangle is deleted once and the triangulation starts out with $O(n)$ separating triangles, the total runtime is $O(n)$.
\end{proof}


It is tempting to try to carry out the reduction by simply flipping the edges of $E'$ in some arbitrary order (or perhaps skipping an edge if it is no longer incident to any separating triangle by the time it is considered). However, this strategy may fail as illustrated in \Cref{fig:bad-flip}. Thus, a more careful strategy is needed. Our algorithm flips only edges that bound separating triangles and do not create any separating triangle, thus strictly decreasing the number of  separating triangles with each flip.

\begin{figure}[H]
    \centering
    \includegraphics[scale=.8]{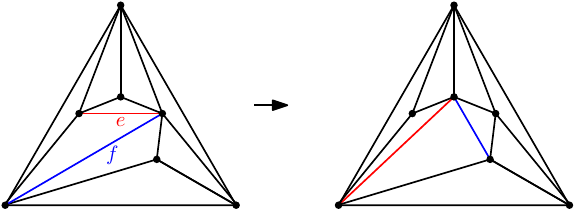}
    \caption{While $\{e,f\}$ is a minimum set of edges hitting every separating triangle, flipping $e$ then $f$ fails to make the triangulation $4$-connected, as a new separating triangle is introduced.}
    \label{fig:bad-flip}
\end{figure}

\begin{lemma}
    \label{lem:bounded-treewidth-flip}
     Given a plane graph $G$ of treewidth $\le k$ with $n$ vertices, and a set $C$ of 3-cycles of $G$, there is a $n\cdot 2^{O(k)}$-time algorithm that computes a minimum cardinality edge set $E'$ such that every 3-cycle in $C$ is incident to at least one edge in $E'$.
\end{lemma}

\begin{proof}
    Recall that a tree decomposition of $G$ is a mapping of $G$ into a tree $T_G$ whose vertices we shall call \emph{bags}.
    A bag $B$ is associated with a subset of vertices $V_B$.
    Every vertex appears in at least one bag and induces a single subtree in $T_G$, and if two vertices are adjacent in $G$ then there is a bag in $T_G$ containing both vertices.
    The width of $T_G$ is the size of its largest bag.
    We root $T_G$ arbitrarily.
    We can build a tree decomposition of $G$ of width $O(k)$ in $n\cdot 2^{O(k)}$ time with the following properties
    \cite{bodlaender1996efficient, kloks1994treewidth}:

\begin{itemize}
    \item the number of bags is $O(n)$;
    \item the maximum degree of a bag is 3;
    \item for a bag with two children, all three bags have identical subsets of vertices; and
    \item for a bag with a single child the two subsets differ by one.
\end{itemize}

Since $G$ is planar, each bag defines an induced subgraph with $O(k)$ edges.
Note that, by definition, if $G$ has a clique, these vertices must appear in a bag of $T_G$.
Thus, every 3-cycle in $C$ must appear in at least one bag of $T_G$.
We compute $E'$ using dynamic programming as follows.
We say that a subset $S$ of edges \emph{satisfies} a subtree of $T_G$ if every 3-cycle of $C$ that appears in a bag of the subtree contains at least one edge in $S$.
For each bag $B$ and for each subset $S$ of the $O(k)$ edges in $V_B$, we define a subproblem $A_{B,S}$ of computing the minimum cardinality set $S'$ of edges satisfying the subtree rooted at $B$ subject to choosing the edges in $S$ in the induced subgraph of $V_B$. 
Thus there are $n \cdot 2^{O(k)}$ total subproblems. By our degree-3 assumption, the children of parents of high-degree simply receive the information from their parent without making any choices.
The child of a low degree bag has to make choices about the incident edges to the potential new vertex $v$. 
The degree of $v$ in the induced subgraph is $O(k)$ which implies that we have to make $2^{O(k)}$ binary choices. Some of these choices are fixed: if the addition of $v$ closes a 3-cycle in $C$, the cycle has two incident edges to $v$ and the choices are valid only if at least one edge of the separating triangle is chosen. 
Therefore to compute each subproblem there is $k\cdot 2^{O(k)}=2^{O(k)}$ nonrecursive work. 
Since the tree decomposition has $O(n)$ nodes, this approach leads to a $n\cdot 2^{O(k)}$ runtime.
\end{proof}

We can now give the proof of \Cref{thm:PTAS}.

\begin{proof}[Proof of \Cref{thm:PTAS}]
The idea is to apply Baker's layer shifting technique~\cite{baker1994approximation} to the hitting set formulation of the problem. We actually solve the more general problem of hitting an arbitrary subset $C$ of $3$-cycles of $G$, given as an input.

Assume without loss of generality that $\epsilon \leq 1/2$ and set $k = \ceil{1/\varepsilon}$. Choose an arbitrary vertex $v_0$ of $G$ and label all vertices of $G$ according to their (shortest path) distance to $v_0$. We call the set of vertices at a specific distance from $v_0$ a layer of $G$. For $0\leq i < k$, assign the color $i$ to all edges incident to two vertices labeled $(i \text{ mod } k)$ (edges joining two vertices in consecutive layers do not get assigned a color).

Let $E^*$ be a smallest subset of edges such that every $3$-cycle in $C$ is incident to at least one edge in $E^*$. Notice that there is at least one of the $k$ colors, say, $c$, for which at most $|E^*|/k$ edges from $E^*$ get assigned the color $c$.

Now imagine decomposing $G$ into the subgraphs $G_0,\ldots G_p$ each induced by $k+1$ consecutive layers, starting at a layer whose label is  $(c \text{ mod } k)$ (except eventually for the first subgraph which starts at layer $0$, and the last subgraph which might consist of fewer than $k+1$ consecutive levels). Note that consecutive subgraphs overlap at layers $(c \text{ mod } k)$. Every $3$-cycle in $C$ appears in one of the subgraphs, as it necessarily has its $3$ vertices in at most $2$ consecutive layers, and every pair of consecutive layers appears in one of the subgraphs. Taking the union of optimal solutions for the problem on each subgraph (with $C$ restricted to the $3$-cycles which appear in that subgraph) thus gives a valid solution for $G$. Let $E_1,\ldots,E_p$ be such optimal solutions, and for $1\leq j \leq p$ let $E^*[G_j]$ denote the subset of edges in $E^*$ which appear in $G_j$. We have that 
\begin{alignat*}{3}
    \left|\bigcup_{1\leq j \leq p} E_j\right| &\leq \sum_{1\leq j \leq p} |E_j|
    &&\leq \sum_{1\leq j \leq p} |E^*[G_j]|\\
    &\leq |E^*| + |E^*|/k
    &&\leq (1+\epsilon)|E^*|,
\end{alignat*}
where the second inequality stems from the fact that $E^*[G_j]$ is a valid solution for $G_j$, the third from the fact that only the edges of color $c$ get counted twice (and there are at most $|E^*|/k$ such edges) and the last from the choice of $k = \ceil{1/\epsilon}$.

The graphs $G_0,\ldots,G_p$ are all $(k+1)$-outerplanar and thus have treewidth $O(k)$ \cite{BODLAENDER19981}. Using \cref{lem:bounded-treewidth-flip},
we can thus solve the problem optimally for each $G_j$ in $|V(G_j)|2^{O(1/\epsilon)}$ time. 
The total runtime is then $\sum_{1\le j\le p}|V(G_j)|2^{O(1/\epsilon)} = n\cdot 2^{O(1/\epsilon)}$.
We can carry out the described procedure for each value of $1\leq c \leq k$ (note that we don't know the right value of $c$ to choose \textit{a priori}) and take the best solution. 
This adds a factor of $1/\epsilon$ that gets absorbed by the exponential.
The previous discussion, together with Lemma \ref{lemma:ptas-reduction}, then gives the sought result.
\end{proof}

While Baker's method of decomposing a planar graph into overlapping layers of small treewidth and combining the solutions is classical, it is interesting to note that a direct adaptation fails for the problem of approximating a smallest set of edges which make a triangulation $4$-connected when flipped simultaneously (or the equivalent problem of augmenting a triangulation to $1$-planar $4$-connected via edge insertions). This is because, while this problem is still efficiently solvable for small treewidth, taking the union of solutions might lead to selecting edges which are not simultaneously flippable. Perhaps this could be overcome by combining the solutions more carefully, but we leave the problem of finding a PTAS for the simultaneous case open here.

\subparagraph{Previous results.} 
Accornero, Ancona and Varini~\cite{AccorneroAV00} present an algorithm and claim it computes the minimum number of edges hitting all separating triangles in a given triangulation.
Note that this result together with our reduction in \cref{sec:hardness} (\cref{cor:flip-hardness}) would imply P=NP.
Although their algorithm computes a set of edges that hits all separating cycles, we claim this set is not minimum.
We attempt here to give a short intuition for that.
In short, their algorithm computes (the equivalent of) a 4-block tree of the input triangulation $T$, a decomposition of $T$ into maximal 4-connected components (4-blocks) hierarchically organized as the nesting structure of separating triangles that bound the maximal components.
At each level, the edges are weighted with the cost of the optimal solution for the children nodes that contain said edge in their outer face.
Note that every edge bounds at most two separating triangles in the next level of the tree and, therefore, there are at most two children associated with the same edge.
At the root of the tree (and recursively in every node) they then compute the minimum weight subset of edges that hit every separating triangle in the maximal 4-block associated with the node.
In~\cite[Theorem 3]{AccorneroAV00} there is an implicit assumption that the optimal solution contains a single edge of each separating triangle, since the weights are computed based on the cost of recursive solutions that flip a specific edge.
That is not true, for example, in our reduction where the optimal solution may flip two edges of a separating triangle bounding the literal gadget (blue edges in \cref{fig:3-4literalGadgets-500}).

%% file: traingulations.tex
\section{\texorpdfstring{$2\to 3$}{2 to 3} Augmentation for Geometric Triangulations  of \texorpdfstring{$n$}{n} Points in Convex Position}
\label{sec:dp}

Given a triangulation $G=(V,E)$ of $n$ points in convex position and a positive integer $\ell$, we can compute the \emph{minimum} number of edges that augment $G$ to a 3- or 4-connected $\ell$-planar graph (or report that there is no such augmentation) by dynamic programming. 

The key observation is a simple combinatorial characterization of 3- or 4-connected augmentation. We introduce some terminology. An edge $e\in E$ is a \emph{diagonal} if it is not an edge of the convex hull of $V$; and $e\in E$ is an \emph{ear} if there is exactly one point in $V$ in an open halfplane bounded by the line spanned by $e$.

\begin{proposition}\label{pro:comb}
    Let $G=(V,E)$ be a straight-line triangulation on $n$ points in convex position, and let $G'=(V,E\cup F)$ be a geometric graph. Then
    \begin{itemize}
        \item $G'$ is 3-connected if and only if every diagonal in $E$ crosses an edge in $F$;
        \item $G'$ is 4-connected if and only if every diagonal in $E$ crosses at least two edges in $F$; and every diagonal in $E$ is either an ear or crosses at least two disjoint edges in $F$. 
    \end{itemize}
\end{proposition}
\begin{proof}
    If $G'$ is 3-connected, then the two endpoints of a diagonal $e\in E$ cannot be a 2-cut, and so some edge $f\in F$ connects the two components of $G-ab$. Since the vertices are in convex position, then $e$ and $f$ cross. 
    
    Conversely, suppose that every diagonal in $E$ crosses an edge in $F$, and $G'$ has a 2-cut $\{a,b\}$.
    Then $\{a,b\}$ is already a 2-cut in $G$, so $ab$ is a diagonal, and $G-ab$ has exactly two connected components. However, an edge $f\in F$ that crosses $ab$ connects the two components, contradicting the assumption that $\{a,b\}$ is a 2-cut in $G'$.

    If $G'$ is 4-connected, then the two endpoints of a diagonal $ab\in E$ and an endpoint of a crossing edge $cd\in F$ cannot form a 3-cut. This leaves us with two possibilities: either $c$ or $d$ is the only vertex in one of the components of $G-ab$ (in this case $ab$ is an ear), or there an edge in $f_1\in F$ that crosses $ab$ and not incident to $c$, and an edge $f_2\in F$ that crosses $ab$ and not incident to $d$. Now if $f_1=f_2$, then $cd$ and $f_1=f_2$ are disjoint, otherwise $f_1$ and $f_2$ are disjoint. 
    In both cases, two disjoint edges in $F$ cross $ab$.
    
    Conversely, suppose that $G'$ satisfies the conditions above but has a 3-cut $\{a,b,c\}$. If none of $ab$, $ac$ and $bc$ is diagonal in $E$, then $G-\{a,b,c\}$ is connected, so $\{a,b,c\}$ would not be a 3-cut. 
    First, assume that $ab$, $ac$, and $bc$ are all diagonal in $G$. Then $G-\{a,b,c\}$ has exactly three components. However each edge of $\Delta(a,b,c)$ crosses two edges in $F$, and in particular at least one edge that is not incident to $\{a,b,c\}$. 
    Thus, each of the three components of $G-\{a,b,c\}$ is connected to another component in $G'$, and so $G'-\{a,b,c\}$ is connected.
    
    Next, assume that only one edge of $\Delta(a,b,c)$, say $ab$, is a diagonal of $E$. Then $G-\{a,b,c\}$ has exactly two components, separated by the line spanned by $ab$. In particular, if $ab$ is an ear, then $c$ cannot be the only vertex on one side of the line spanned by $ab$. In this case, $ab$ crosses at least two edges $f_1,f_2\in F$. At most one of them is incident to $c$: This is clear if $ab$ is not an ear, and $f_1,f_2$ are disjoint. If $ab$ is an ear, then $f_1,f_2$ are incident to a single vertex on one side of $ab$, but this vertex is not $c$, and so one of $f_1$ and $f_2$ is not incident to $c$. Since $F$ contains an edge between the two components of $G-\{a,b,c\}$, then $G'$ is connected.
\end{proof}

\begin{theorem}\label{thm:dp}
    For any $k\in \{3,4\}$ and any constant $\ell\in \mathbb{N}$, there is an  algorithm that, given a PSLG triangulation $G=(V,E)$ on $n$ points in convex position, can find a minimum set $F$ of new edges in $O(n)$ time such that $G'=(V,E\cup F)$ is a $k$-connected $\ell$-planar geometric graph, or reports that there is no such graph $G'$. 
\end{theorem}
\begin{proof}
Let $e_0$ be an arbitrary edge of the convex hull of $V$, and let $E_0$ be the subset of $E$ comprising $e_0$ and all diagonals in $E$. We define a partial order $\preceq $ on $E_0$: We say that $e\preceq f$ if both $f$ and $e_0$ are in the same closed halfplane of bounded by the line spanned by $e$. (In particular, $e\preceq e_0$ holds for all $e\in E_0$. The poset $(E_0,\preceq)$ defines a tree, rooted at $e_0$, which is binary tree rooted at $e_0$. 
For a dynamic programming algorithm, we define the subproblems that correspond to the edge $e\in E_0$.

\begin{figure}[htbp]
\hfill
     \begin{subfigure}[b]{0.25\textwidth}
          \centering
        \includegraphics{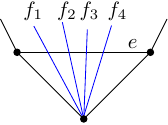}
    \subcaption{}
    \label{fig:dp1}
     \end{subfigure}
  \centering
     \begin{subfigure}[b]{0.36\textwidth}
          \centering
  \includegraphics{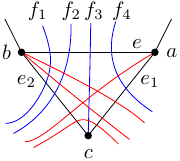}
    \subcaption{}
    \label{fig:dp2}
     \end{subfigure}
     \hfill
     \begin{subfigure}[b]{0.36\textwidth}
          \centering
   \includegraphics{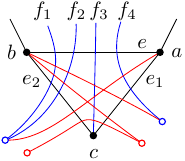}
      \subcaption{}
    \label{fig:dp3}
     \end{subfigure}
       \caption{Illustration for a dynamic programming algorithm: 
       An ear diagonal $e$ (a); a non-ear diagonal $e$ for $k=3$ (b), and for $k=4$ (c). 
       \label{fig:dp}} 
       \end{figure}

\subparagraph{3-connectivity.}
For an edge $e\in E_0$, integer $q\in \{0,\ldots ,\ell\}$ and a vector $\vec{c}\in \{1,\ldots , \ell\}^q$, let $\opt_3(e,q,\vec{c})$ denote the minimum size of a set of new edges $F$ such that $G'=(V,E\cup F)$ has the following properties: 
\begin{itemize}
    \item Edge $e$ crosses exactly $q$ edges in $F$: $f_1,\ldots ,f_q$; 
    \item for $i=1,\ldots , q$, edge $f_i\in F$ crosses at most $c_i$ edges $h\in E_0$, $h\preceq e$; and 
    \item every edge $h\in E_0$, $h\preceq e$, crosses an edge in $F$ (cf.~\Cref{pro:comb} for $k=3$).
\end{itemize}
If no such set $F$ exists, then we set $\opt_3(e,q,vec{c})=\infty$. 

Our DP algorithm considers all edges $e\in E_0$ in increasing order in the poset $(E_0,\preceq)$. 
If $e$ is a minimal element of the poset $(E_0,\preceq)$, then $e$ is an ear, and $\opt_3(e,q,\vec{c})=q$ for all $\vec{c}$; see \Cref{fig:dp1}.
If $e$ is not minimal, then we may assume that $e=ab$ in  a triangle $\Delta(a,b,c)$, where $e_1=ac$ or $e_2=bc$ is also a diagonal in $E_0$. To compute $\opt_3(e,q,\vec{c})$, we compare all possibilities in the triangle $\Delta(a,b,c)$ by brute force: Each of the $q$ edges that cross $e$ may cross $e_1$ or $e_2$, or end at vertex $c$. 
Furthermore, any edge that crosses $e_1$ (resp., $e_2$) may end at $b$ (resp., $a$) or cross $e$ or $e_2$ (resp., $e_1$). For each combination, we determine whether they are feasible (i.e., satisfy all constraints), and let 
$\opt_3(e,p,\vec{c})$ be the minimum size of a feasible solution; see \Cref{fig:dp2}.
Importantly, the order in which $e$ crosses $f_1,\ldots , f_q$ is not specified---we count only unavoidable  crossings: The edges in $\{f_1,\ldots f_q\}$ that are incident to $c$ must cross all edges that cross both $e_1$ and $e_2$, or cross $e_1$ and end at $b$, or cross $e_2$ and end at $a$. Similar conditions hold for edges that cross $e_1$ or $e_2$ by symmetry.

\subparagraph{4-connectivity.} In this case, the combinatorial characterization in \Cref{pro:comb} is more involved, and we need to maintain information for bundles of new edges that cross an old edge. For an edge $e\in E_0$, integer $q\in \{0,\ldots ,\ell\}$, vector $\vec{c}\in \{1,\ldots , \ell\}^q$, and a laminar set system $P$ over $\{1,\ldots, q\}$, let $\opt_4(e,q,\vec{c},P)$ denote the minimum size of a set of new edges $F$ such that $G'=(V,E\cup F)$ has the following properties: 
\begin{itemize}
    \item Edge $e$ crosses exactly $q$ edges in $F$: $f_1,\ldots ,f_p$; 
    \item for $i=1,\ldots , p$, edge $f_i\in F$ crosses at most $c_i$ edges $h\in E_0$, $h\preceq e$; 
    \item each set $S\in P$ corresponds to a 
    set $F_{e,S}=\{f_i:i\in S\}$ such that the edges in $F_{e,S}$ have a common endpoint below $e$ or $F_{e,S}$ is the set of all edges in $F$ that cross an edge $h\preceq e$; 
    \item  every edge $h\in E_0$, $h\preceq e$, satisfies the conditions in \Cref{pro:comb} for $k=4$
    assuming the edges in each set $F_{e,S}$, $S\in P$, do not have a common endpoint \emph{above} $e$.
    \end{itemize}
If no such set $F$ exists, then we set $\opt_4(e,q,\vec{c},P)=\infty$.

Our DP algorithm considers all edges $e\in E_0$ in increasing order in the poset $(E_0,\preceq)$. 
If $e$ is a minimal element of the poset $(E_0,\preceq)$, then $e$ is an ear. In this case, all edges that cross $e$ must have a common endpoint below $e$: We set $\opt_3(e,q,\vec{c},P)=q$ for all $\vec{c}$ if $P=\{\{1,\ldots q\}\}$, and $\opt_3(e,q,\vec{c},P)=\infty$ for any other partition $P$; see \Cref{fig:dp1}.
If $e$ is not minimal, then $e=ab$ in a triangle $\Delta(a,b,c)$, where $e_1=ac$ or $e_2=bc$ is also a diagonal in $E_0$. We compute  $\opt_4(e,q,\vec{c},P)$ by a brute force comparison of all possible subproblems for edges $e_1$ and $e_2$. By maintaining the laminar system $P$, we can determine whether every non-ear edge is crossed by at least two disjoint edges. 
\end{proof}

%% file: pslg.tex
\section{Augmentation for Planar Straight-Line Graphs}
\label{sec:pslg}

As noted above, the connectivity augmentation problem has been thoroughly studied over PSLGs, where both the input and output graphs are restricted to PSLGs (hence $k\leq 5$).
For example, every PSLG $G$ can be augmented to a 2-connected PSLG, and this is the best possible if the vertices are in convex position or an edge of $G$ is a chord of the convex hull. 
In this section, we relax the planarity constraint for the output graph $G'$. In general, we would like to  determine all pairs of parameters $(k,\ell)$ such that every PSLG can be augmented to a $k$-connected $\ell$-planar geometric graph. 

If we wish to augment a PSLG to a 3-connected geometric graph, we need to go beyond planarity. Our first result addresses the case $k=3$ for points in convex position.

\begin{restatable}{theorem}{pslgtheorem}
\label{thm:1}
Every PSLG on $n$ points in convex position can be augmented to a 3-connected 5-planar geometric graph. This bound is the best possible: There exists a triangulation on $n$ points in convex position for which any augmentation to a 3-connected geometric graph yields an edge with at least 5 crossings. 
\end{restatable}
\begin{proof}
\textbf{Upper bound.} 
Let $G=(V,E)$ be a PSLG on $n$ points in convex position. Assume w.l.o.g.\ that $G$ is an edge-maximal PSLG (i.e., a triangulation of a convex $n$-gon). Then $G$ is 2-connected; and a vertex pair $\{a,b\}$ is a 2-cut if and only if $ab\in E$ and $ab$ is a diagonal of the outer cycle. Let us call such an edge a \emph{diagonal}. Augment $G$ as follows: For every diagonal $ab\in E$, consider the two adjacent triangles, say $\Delta(abc)$ and $\Delta(abd)$, and insert the edge $cd$. Let $G'$ denote  the resulting geometric graph. First, we show that $G'$ is 3-connected: Indeed, if $\{a,b\}$ is a 2-cut in $G'$, then it is already a 2-cut in $G$. However, $G-\{ab\}$ has exactly two components, which are connected by an edge in $G$ (by construction). 
Second, we show that $G'$ is 5-planar. Every edge in $E$ crosses exactly one (new) edge by construction. Consider a new edge, say $cd$, with the notation above. Then $cd$ crosses one old edge (edge $ab$). It can cross at most two (new) edges in $\Delta(abc)$, and at most two (new) edges to $\Delta(abd)$. Overall, every edge of $G'$ crosses at most 5 other edges.

\smallskip\noindent\textbf{Lower bound.}
Let $G=(V,E)$ be a balanced triangulation on $n$ points in convex position, where $n$ is sufficiently large. That is, the dual graph is a balanced binary tree, where the root $\Delta_0$ has degree 3. For $k\in \mathbb{N}$, denote by $B_0(k)$ the union of triangles at distance at most $k$ from $\Delta_0$ in the dual graph. If $n\geq 3\cdot 2^{11}$, then all triangles in $B_0(10)$ have degree 3. 

Suppose, for the sake of contradiction, that $G$ can be augmented to a 3-connected 4-planar geometric graph $G'=(V,E\cup F)$. Since $G'$ is 3-connected, then every interior edge of $G$ is crossed by at least one edge in $F$.  Let us define the \emph{length} of an edge $e\in F$ as the number of faces of $G$ that $e$ intersects. Since $G$ is a triangulation of a convex polygon, then every edge in $F$ has length at least 2. By assumption, every new edge crosses at most 4 old edges, and so the length of every edge in $F$ is at most 5. 


We proceed with five claims to conclude that a new edge in $G'$ has at least 5 crossings.

\begin{figure}[htbp]
\hfill
     \begin{subfigure}[b]{0.32\textwidth}
          \centering
         \includegraphics[width=45mm]{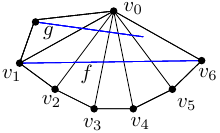}
    \subcaption{}
    \label{fig:claim1}
     \end{subfigure}
  \centering
     \begin{subfigure}[b]{0.32\textwidth}
          \centering
   \includegraphics[width=45mm]{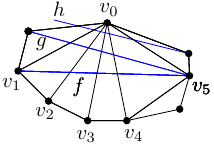}
    \subcaption{}
    \label{fig:claim2}
     \end{subfigure}
     \hfill
     \begin{subfigure}[b]{0.32\textwidth}
          \centering
   \includegraphics[width=45mm]{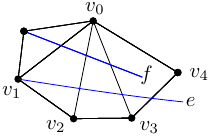}
      \subcaption{}
    \label{fig:claim3}
     \end{subfigure}
       \caption{Illustrations for \Cref{cl:1} (a), 
       \Cref{cl:2} (b), and \cref{cl:3} (c). 
       \label{fig:claims}} 
       \end{figure}

\begin{claim}\label{cl:1}
$B_0(10)$ does not contain any edge $f\in F$ of length 5 such that the 5 triangles that intersect $f$ have a common vertex. 
\end{claim}
\begin{proof}[Proof of Claim~\ref{cl:1}]
Suppose that $f\in F$ crosses four triangles $\Delta(v_0,v_i,v_{i+1})$ for $i=1,\ldots , 5$; see \Cref{fig:claim1}. Since all these triangles have dual degree three, then $v_0 v_1$ is an interior edge, which are each crossed by some edge $g\in F$. Edge $f$ crosses 4 old edges, so it cannot cross any new edge. Edge $g$ crosses $v_0v_1$, so it is not incident to $v_0$. It cannot cross $f$, then it must cross $v_0v_i$ for $i=1,\ldots , 5$.  It crosses at least 5 old edges: a contradiction. 
\end{proof}

\begin{claim}\label{cl:2}
$B_0(10)$ does not contain any edge $f\in F$ of length 4 or 5 such that the first 4 triangles traversed by $f$ have a common vertex. 
\end{claim}
\begin{proof}[Proof of Claim~\ref{cl:2}]
Suppose that $f\in F$ crosses four triangles $\Delta(v_0,v_i,v_{i_+1})$ for $i=1,\ldots , 4$; see \Cref{fig:claim2}. By Claim~1, we know that $f$ does not cross $v_0v_5$. Since these triangles have dual degree three, then $v_0 v_1$ and $v_0v_5$ are interior edges, each of which is crossed by some edges $g,h\in F$. First, we show that $g\neq h.$ Indeed, if $g=h$, then this edge either crosses $v_0v_i$ for $i=1,\ldots , 5$, or it crosses $f$ and four boundary edges of the region $\bigcup_{i=1}^4 \Delta(v_0v_iv_{i+1})$: both cases lead to contradiction, so $f\neq g$. 
Now edge $f$ crosses at least 3 old edges, so it can cross at most one new edge. Edges $g$ and $h$ are not incident to $v_0$; at most one of them can cross $f$.
Assume w.l.o.g.\ that $g$ does not cross $f$. Then $g$ must cross $v_0v_i$ for $i=1,\ldots , 4$. 
If $g$ crosses $v_0v_5$ or $h$, as well, then it crosses at least 5 edges: a contradiction. 
Therefore, $g$ crosses neither $v_0v_5$ nor $h$. Then $g$ must be incident to $v_5$. In this case, however, $h$ crosses $v_0v_i$ for $i=1,\ldots, 5$: a contradiction.
\end{proof}

\begin{claim}\label{cl:3}
Let $e\in F$ be an edge of length at least 3 such that the first three triangles traversed by $e$ are $\Delta(v_0,v_i,v_{i+1})\subset B_0(10)$ for $i=1,2, 3$; and $e$ is incident to $v_1$. Then there is an edge $f\in F$ that crosses both $v_0v_1$ and $e$. 
\end{claim}
\begin{proof}[Proof of Claim~\ref{cl:3}]
Since $e$ is in $B_0(10)$, then $v_0v_1$ is an interior edge, crossed by some edge $f$; see \Cref{fig:claim3}. If $f$ crosses $e$, then our proof is complete. 
Otherwise, $f$ crosses $v_0 v_i$ for $i=1,2,3$. Which contradicts Claim~\ref{cl:2}. 
\end{proof}

\begin{claim}\label{cl:4}
$B_0(10)$ does not contain any edge of length 4 or 5.
\end{claim}
\begin{proof}[Proof of Claim~\ref{cl:4}] 
Let $e\in F$ be a new edge contained in $B_0(10)$. Note that any three consecutive triangles traversed by $e$ have a common endpoint. If $e$ has length 5, then it crosses 4 old edges, and it also crosses a new edge by \Cref{cl:3}.
If $e$ has length 4, then it crosses 3 old edges. By \Cref{cl:3}, it cannot traverse 4 triangles with a common vertex. Therefore, by \Cref{cl:2}, the first and last triangles each yield new edges $g,h\in F$ that cross $e$. (Note that $g\neq h$, otherwise the edge $g=h$ would cross 5 old edges.)
\end{proof}

\begin{figure}[htbp]
\hfill
     \begin{subfigure}[b]{0.32\textwidth}
          \centering
         \includegraphics[width=45mm]{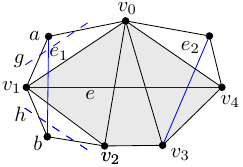}
    \subcaption{}
    \label{fig:claim5a}
     \end{subfigure}
  \centering
     \begin{subfigure}[b]{0.32\textwidth}
          \centering
   \includegraphics[width=45mm]{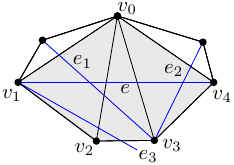}
    \subcaption{}
    \label{fig:claim5b}
     \end{subfigure}
     \hfill
     \begin{subfigure}[b]{0.32\textwidth}
          \centering
   \includegraphics[width=45mm]{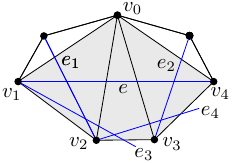}
      \subcaption{}
    \label{fig:claim3c}
     \end{subfigure}
       \caption{Illustrations for \Cref{cl:5}. 
       \label{fig:claim5}} 
       \end{figure}

\begin{claim}\label{cl:5}
$B_0(10)$ does not contain any edge of length 3 or more.
\end{claim}
\begin{proof}[Proof of Claim~\ref{cl:5}] 
Let $e\in F$ be an edge of length 3 that traverses triangles $\Delta(v_0v_iv_{i+1})$ for $i=1,2,3$. The endpoints of $e$ are $v_1$ and $v_4$; and $e$ crosses two old edges; see \Cref{fig:claim5}. \Cref{cl:3} yields an edge $e_1$ (resp., $e_2$) that crosses $v_0v_1$ (resp., $v_0v_4)$ and $e$. By \Cref{cl:1}, $e_1\neq e_2$, so $e$ crosses at least two new edges. Consequently, no other edge can cross $e$. 

Consider the pentagon $P=\bigcup_{i=1}^3 \Delta(v_0v_iv_{i+1})$. It contains the edge $e$ of length 3. Edge $e_1$ enters the interior of $P$ by crossing edge $v_0v_1$. We claim that $e_1$ cannot exit $P$ by crossing $v_1v_2$. Suppose, to the contrary, that $e_1$ crosses $v_1v_2$; see \Cref{fig:claim5a}. Let $\Delta(a,v_0,v_1)$ and $\Delta(b,v_1,v_2)$ be the triangles adjacent to $P$.
Since the length of $e_1$ is 3 by \Cref{cl:4}, then it is incident to $a$ and $b$. By \Cref{cl:3}, $e_1$ crosses new edges $g$ and $h$ that crosses $v_1a$ and $v_1b$, respectively. Note that $g\neq h$ (or else the edge $g=h$ would have length 5, contrary to \Cref{cl:4}). Overall $e_1$ crosses 5 edges---a contradiction. If $e_1$ exits on any other edge of $P$, then it would have length at least 4, contradicting \Cref{cl:4}. We conclude that $e_1$ crosses the boundary of $P$ only once, so it ends in $P$. The same argument holds for $e_2$. 
 
Since $v_2v_3$ is a diagonal, there is an edge $e_3\in F$ that crosses $v_2 v_3$; see \Cref{fig:claim5b}. As argued in the previous paragraph, $e_3\notin \{e_1,e_2\}$. Edge $e_3$ cannot cross $e$ (as $e$ already has 4 crossings), and its length is at most 3 by \Cref{cl:4}. So $e_3$ cannot exit the pentagon $P$, and it ends at $v_1$ or $v_4$. Assume w.l.o.g.\ that $e_3$ ends at $v_1$. 

Since $v_3v_4$ is also a diagonal, there is an edge $e_4\in F$ that crosses $v_3v_4$. By the previous two paragraphs, $e_4\not\in \{e_1,e_2,e_3 \}$. The length of $e_4$ is at most 3 by \Cref{cl:4}, so it also ends in $P$: at vertex $v_0$ or $v_2$. Since $e_4$ cannot cross $e$, either, then $e_4$ ends at vertex $v_2$. Consequently, $e_3$ and $e_4$ cross. 

Recall that $e_1$ crosses $v_0v_1$ and edge $e$, and ends at $v_2$ or $v_3$. If $e_1$ ends at $v_2$, then it crosses $e_3$. In this case, $e_3$ has five crossings: It crosses two old edges, and three new edges: $e_1$, $e_4$ and a 3rd edge that crosses $v_1v_2$ by \Cref{cl:3}.
If $e_1$ ends at $v_3$, then it crosses $e_4$. 
In this case $e_4$ has five crossings: It crosses two old edges, and three new edges: $e_1$, $e_2$, and $e_3$. Both cases lead to a contradiction. 
\end{proof}

This completes the proof of Theorem~\ref{thm:1}.
\end{proof}

We can generalize Theorem~\ref{thm:1} and obtain an asymptotically tight trade-off between connectivity and local crossing number (which matches the bound in \Cref{pp:general}). 

\begin{theorem}\label{thm:cxpslg}
 For every $k\in \mathbb{N}$, every PSLG on $n$ points in convex position can be augmented to a $k$-connected $O(k^2)$-planar geometric graph; and this bound is the best possible.
\end{theorem}
\begin{proof} 
We may assume w.l.o.g.\ that the input is a triangulation $G$ on $n$ points in convex position. The dual graph is a binary tree with $n$ nodes. We can partition the dual tree into subtrees (\emph{clusters}) of size in the range $[k,3k]$. Each cluster is a union of triangles, i.e., a convex polygon with $\Theta(k)$ vertices. The adjacency graph of the clusters (\emph{cluster graph}) is a tree of maximum degree $O(k)$. 
  
We can augment $G$ as follows. (1) Augment each cluster to a complete graph; (2) for each pair of adjacent clusters $V_1$ and $V_2$ (where $|V_1\cap V_2|=2$) add a matching of size $k-2$ between $V_1\setminus V_2$ and $V_2\setminus V_1$. Let $G'$ be the augmented graph. 
Similarly to the proof of \Cref{thm:topological}, the geometric graph graph $G'$ is $k$-connected (by induction on the clusters), and $O(k^2)$-planar (arguing about edges induced by a cluster, and a pair of adjacent clusters). 
\end{proof}

For $n$ points in general position, we cannot always augment a PSLG to 4-connectivity with  bounded local crossing number.

\begin{theorem}\label{thm:fan}
For every $n\in \mathbb{N}$, there is a straight-line triangulation $G$ on $n$ points such that any augmentation to a 4-connected geometric graph has an edge with $\Omega(n)$ crossings. 
\end{theorem}  
\begin{figure}[htbp]
  \centering
    \includegraphics[width=0.5\textwidth]{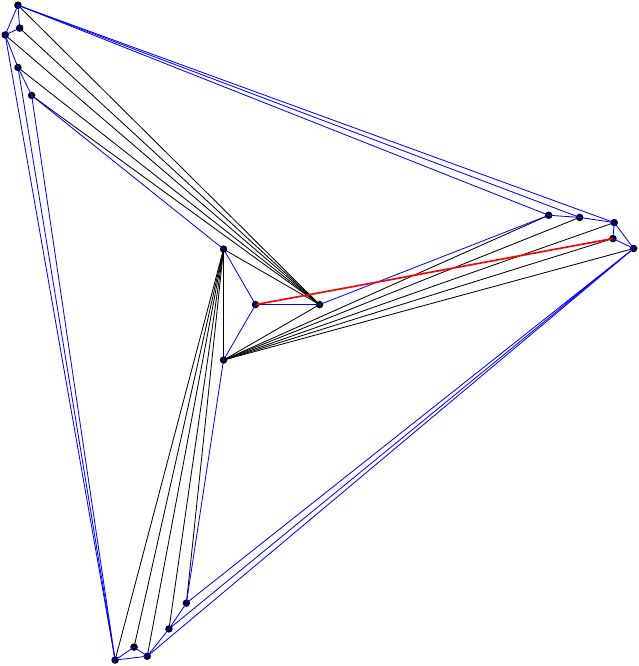}
    \caption{Construction in the proof of \Cref{thm:fan}}
    \label{fig:fan}
\end{figure}
\begin{proof} 
We construct a straight-line triangulation as follows; see \Cref{fig:fan}. Start with a PSLG $K_4$, with three points at the vertices of a regular triangle, and one at the center. Attach $(n-4)/3$ long edges to each outer vertex to form a ``fan'' almost parallel to the three sides of the regular triangle. Let $G$ be an arbitrary triangulation of this graph. In any 4-connected augmentation $G'$ of $G$, the central vertex of the initial $K_4$ must be connected to a new vertex. This edge will cross $\Omega(n)$ edges of one of the fans, hence $\Omega(n)$ edges of $G$.
\end{proof}